%% file: main.tex
\documentclass[
aip,
jcp,
reprint, 
floatfix
]{revtex4-1}

\usepackage{graphicx}
\usepackage{dcolumn}
\usepackage{xcolor}
\usepackage{bm}
\usepackage{hyperref}
\usepackage{ifthen}
\usepackage{amsmath,amssymb,amsthm}
\usepackage{fancybox}
\usepackage{titlesec}
\usepackage{tikz}
\usepackage{subcaption} 
\usepackage[ruled, vlined]{algorithm2e}
\usepackage{placeins} 
\usetikzlibrary{arrows.meta, calc}

\usepackage[color=yellow!30]{todonotes}

\newtheorem{theorem}{Theorem}[section]

\newtheorem{proposition}[theorem]{Proposition}

\makeatletter
\def\@email#1#2{%
 \endgroup
 \patchcmd{\titleblock@produce}
  {\frontmatter@RRAPformat}
  {\frontmatter@RRAPformat{\produce@RRAP{*#1\href{mailto:#2}{#2}}}\frontmatter@RRAPformat}
  {}{}
}%
\makeatother

\begin{document}

\preprint{AIP/123-QED}

\title{Coupling between Phase Separation and Geometry on a Closed Elastic Curve:\\
Free Energy Minimization and Dynamics}

\author{Hanchun Wang}
\author{Ronojoy Adhikari}
\author{Michael E. Cates}
\email[Corresponding author: ]{m.e.cates@damtp.cam.ac.uk}
\affiliation{DAMTP, University of Cambridge, Wilberforce Road, Cambridge CB3 0WA, UK}
\date{\today}
\begin{abstract}
\noindent
We study the free energy and dynamics of a closed elastic filament (a one-dimensional curve in two dimensions) coupled to a scalar concentration field representing, for example, an absorbed species. The density variable has a tendency to phase-separate whereas the local spontaneous curvature is concentration-dependent. We address analytically and by simulation both the free energy landscape and the dynamics (the latter
comprising a coupled Willmore flow and Cahn--Hilliard gradient flow on the full differential geometry of a closed filament), addressing issues that previous work typically sidestepped by restricting to the Monge gauge. Specifically we find that the closure constraint for a deformable filament qualitatively changes the free energy landscape compared with either a rigid closed filament or an open elastic one, admitting metastable and stable states with more than one domain of each type. By numerical global free energy minimization we explore equilibrium morphologies across a wide range of model parameters. For selected parameter values we present fully dynamical results, tracking the time evolution of the various contributions to the free energy and confirming the emergence of both metastable and equilibrium multi-domain morphologies.

\end{abstract}

\maketitle

\section{INTRODUCTION}
In various soft matter and biological settings, phase separation arises within or on a deformable structure, whose shape coevolves with the local composition. The phase-separating species in such systems can both sense and impose local curvature, so that compositionally heterogeneous regions deform the substrate while the substrate curvature reciprocally biases composition. 
For example, lipid bilayers behave as elastic surfaces \cite{Canham1970,helfrich1973elastic,Seifert1997,dimova2006practical}, on which curvature-sensing and curvature-inducing proteins (such as BAR-domain assemblies, endocytic machinery) and lateral lipid demixing together generate compositional domains coupled to membrane curvature \cite{mcmahon2005membrane,Prevost2015,simunovic2016curvature, VeatchKeller2003,SimonsToomre2000,sorre2009curvature,roux2005role, GoychukFrey2019,tozzi2019out,le2021dynamic, bohinc2005self}; this coupling underlies cellular processes such as endocytosis, exocytosis, and protein- or cytoskeleton-driven membrane remodelling \cite{mcmahon2015membrane,salbreux2012actin,kaksonen2018mechanisms,jahn1999membrane,wu2014exocytosis}. 
Within the cytoplasm, biomolecular condensates, formed by liquid--liquid phase separation, have interfaces that need not be governed by surface tension alone: interfacial organization can endow condensate boundaries with additional surface mechanics, and stabilize nonspherical shapes \cite{Brangwynne2009,Hyman2014,Banani2017,folkmann2021regulation,law2023bending}. 
Similar couplings can arise with DNA and chromatin substrates: liquid condensates form along chromatin and exert capillary forces that mechanically restructure the polymer~\cite{shin2018liquid,larson2017liquid, keenen2021hp1}, whereas on plasmid DNA, ParB and its CTP-driven bridging interactions assemble a phase-separated condensate that compacts and bends the deformable substrate~\cite{brackley2013nonspecific,tivsma2023dynamic,zhao2024phase}.
This coupling (as seen below) can stabilize nontrivial stationary compositional patterns, in contrast to phase separation in a fixed spatial domain which generically coarsening toward states of minimal interface cost \cite{10.1063/1.1744102}.

In this paper, we study a closed planar elastic filament (or `one-dimensional membrane'\cite{mark2010physical}) carrying a conserved phase-separating field. Bending and stretching elasticity \cite{helfrich1973elastic,willmore1993riemannian} are coupled to Cahn--Hilliard phase separation \cite{10.1063/1.1744102} through a composition-dependent spontaneous curvature. We show that allowing the filament geometry to deform under a curvature-composition coupling changes the free energy landscape. In particular, we show that the closed-loop setting admits distinct classes of stable and metastable morphologies that are absent on a fixed periodic domain: configurations with multiple domains can persist as local minima. This is crucially different from an open filament with similar coupling, because for a closed loop interface displacements generically require global curvature readjustments to maintain closure. Thus, global topology profoundly affects phase-separation dynamics.

Mechanochemical coupling between composition and curvature has been widely modeled. A common effective description introduces a curvature-density feedback in which an adsorbate density (describing {\em e.g.}, protein enrichment) generates an effective spontaneous curvature, while curvature biases recruitment and transport \cite{markin1981lateral, GoychukFrey2019,tozzi2019out,le2021dynamic, kralj1996shapes}. Many such theories focus on pattern formation and shape instabilities on deformable interfaces, often adopting symmetry-reduced or graph-based geometries ({\em e.g.}, Monge representations) to obtain tractable evolution equations and numerics  \cite{Frey2023mechanochemical,PhysRevE.111.044405,elliott2021domain}. These are restricted to small deformations from a flat or weakly-curved reference state. In contrast, the present work operates directly in the full intrinsic geometry of the closed curve, with no assumption of small deformations. In parallel, elastic-curve mechanics and shape transitions have been studied in continuum settings \cite{nakayama1993motion,taloni2023general,brazda2023bifurcation}. However, these approaches typically do not isolate the role of global constraints specific to closed loops,
which can mediate nonlocal interactions between phase boundaries and restructure the free energy landscape. Notably, as we shall illustrate later ({\em e.g.}, Fig.~\ref{fig:closure_compare_combined}), periodic boundary conditions on the intrinsic curvature and composition fields are not sufficient to enforce geometric closure of a filament. More closely related to our work is the analysis of Kawakatsu et al.~\cite{kawakatsu1993}, who treated a similar planar closed loop geometry with curvature--composition coupling. These authors used a sharp domain-wall, $n$-fold-symmetric, static ansatz that yields only ground-state branches of the free energy corresponding to one or two phase-separated domains. Here we instead retain a full Cahn--Hilliard composition field with diffuse interfaces; in addition we allow for filament extensibility, impose no symmetry, and follow the full gradient-flow dynamics, finding metastable multi-domain morphologies that were previously unexplored.

The time evolution of (noiseless) phase-separation on an elastic manifold is governed by a gradient flow on a suitable free energy functional. (This is not true for active phase separation \cite{Cates_2025}, but here we address only the passive case.) Several works have developed numerical methodologies to address this type of problem. For one-dimensional elastic bodies, Cosserat rod models discretized by finite differences \cite{PyElastica}, and discrete elastic-band formulations \cite{Bergou2008discrete}, have been used. The Arbitrary Lagrangian-Eulerian (ALE) framework has also been applied to evolving interfaces coupled to additional fields \cite{membranes15040106}. Discrete differential geometry is used in the methodology of \cite{choi2024dismech} for both active and passive deformable objects. Geometrical methods involving moving frames also represent a powerful recent methodology~\cite{kikuchi2024rare,nemeth2024geometric}. Surface Cahn--Hilliard formulations and related curvature-driven geometric flows have been analyzed and discretized on moving interfaces using finite elements on a moving surface \cite{elliott2010surface, elliott2021domain}. 

Here, to address the full gradient flow problem without reduced equations, we develop a practical numerical framework for coupled Cahn-Hilliard and Willmore-flow-type\cite{willmore1993riemannian} (overdamped) elastic dynamics on a filament. This represents a 3-field, 4th order, nonlinear coupled PDE system. We first compute metastable equilibria by direct constrained energy minimization in a spectral representation, and then obtain dynamical pathways by integrating the coupled gradient flow PDEs. Using these tools, we map the resulting morphology diagrams across parameter space, identify stable and metastable branches, and relate their transition boundaries to the competition between interfacial cost in the composition field and `closure-induced curvature frustration'. 

Our results demonstrate in detail how such frustration qualitatively changes the free energy landscape from that of either a rigid closed domain or a flexible open one, admitting new metastable morphologies including several classes of multi-interface phase-separated states. These morphologies are separated by first-order morphological transitions and have a strong influence on dynamical evolution, as well as on the phase diagram.

Specifically, closure directly modulates domain merging by coupling the Cahn--Hilliard field to an additional field---the curvature field---which must satisfy the global double-integral closure constraints. In this setting, the standard Cahn-Hilliard model outcome, that coarsening ends in a single domain, is no longer trivially applicable: a one-domain composition profile generically cannot accommodate geometric closure.

We find three primary resolution routes, starting from a configuration comprising a single domain of each type, to comply with the closure constraint. (i) The system can eliminate interfaces, approaching a uniform phase and circular shape at the expense of additional stretching energy. (ii) It can retain a single domain structure but satisfy closure mainly by paying extra bending energy, through a globally frustrated curvature profile. (iii) It can instead introduce two or more extra interfaces, forming a multi-domain state whose added interfacial cost allows compliance with the closure constraint. 

Sharp-interface analysis and numerical simulations then show that, in the nearly sharp-interface regime, the two-domain (``peanut'' shape) morphology is the minimal-interface configuration capable of satisfying closure without paying additional bending energy beyond necessity, while additional interfaces could also accommodate further global constraints (e.g. prescribed enclosed area). In our setting, these mechanisms correspond to the branches analyzed later as $N=0,2,4, 6$ respectively (Fig. \ref{fig: metastables}). The resulting tradeoffs are controlled by the relative costs of stretching, bending, and interfacial energies; varying the parameters that set these costs yields the morphology diagram and phase boundaries reported in Sec. \ref{sec:results}, as confirmed by our simulations and energy comparisons.

The remainder of the paper is organized as follows: in Sec.~\ref{sec: background} we introduce the model, its kinematics, and the coupled gradient-flow dynamics; in Sec.~\ref{sec: analysis} we analyze the sharp-interface and near-inextensible limits to obtain analytical constraints on admissible morphologies; in Sec.~\ref{sec: numerical} we describe the numerical methods for constrained minimization and time-dependent simulations; in Sec. \ref{sec:results} we present the resulting morphology landscapes and phase diagrams. Additional technical details are provided in the Appendices.

\section{BACKGROUND \& MODEL} \label{sec: background}

\subsection{Physical setting and modeling assumptions}\label{sec: assumptions}

\begin{figure*}
    \centering
    \includegraphics[width=0.9\linewidth]{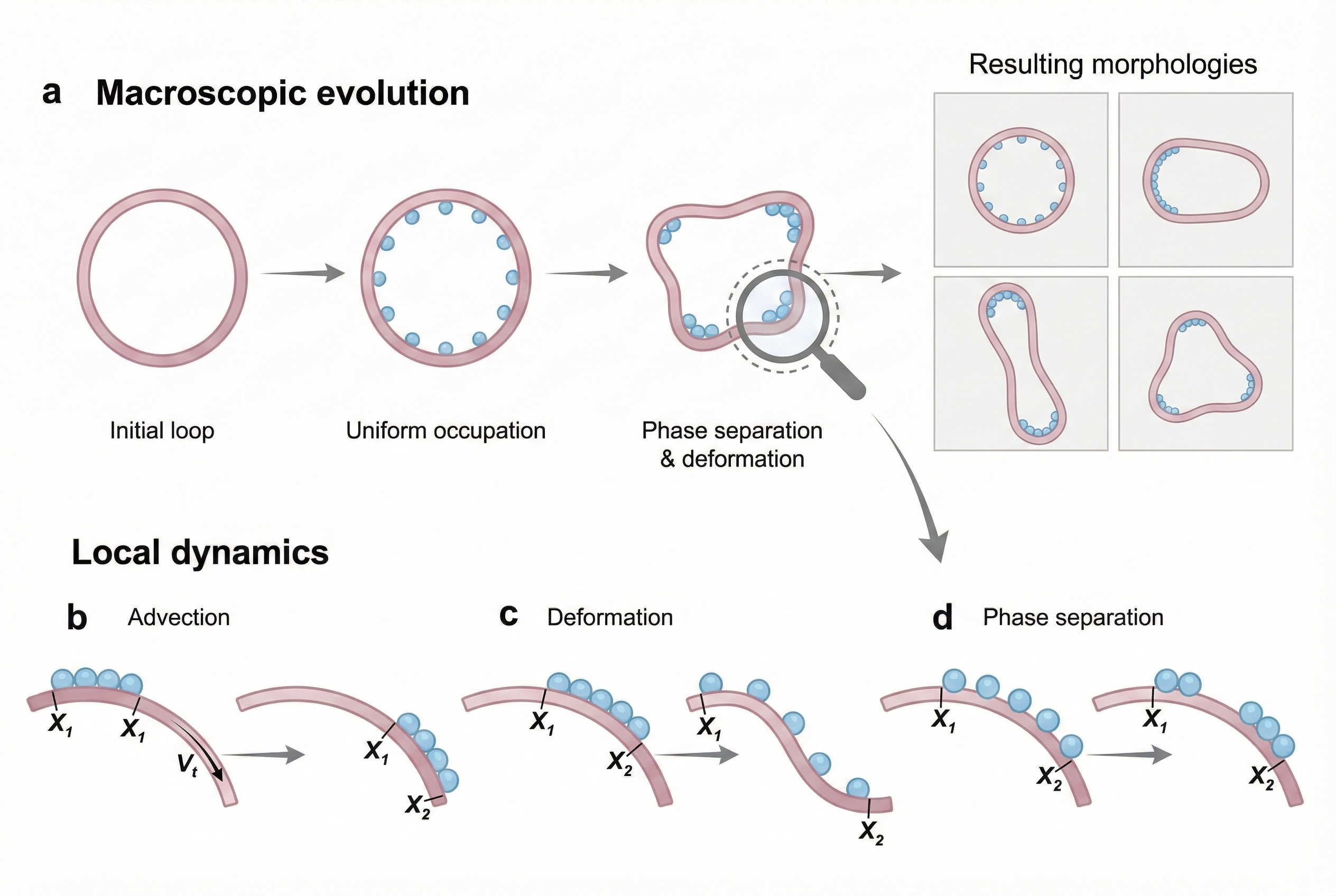}
    \caption{\textbf{Macroscopic shape evolution driven by coupled local dynamics.}
\textbf{a,} \emph{Adsorption:} Adsorbed particles that induce a spontaneous curvature locally bend the filament. Because of the dynamical coupling, such particles collectively drive the global evolution of the closed filament away from the circular shape with the lowest bending energy to different morphologies in which symmetry is broken as phase separation and loop deformation set in concurrently. The right panel shows representative final morphologies featuring different numbers of domains and distinct shapes.
\textbf{b--d,} Local dynamical mechanisms governing the evolution, illustrated between arc-length positions $X_1$ and $X_2$.
\textbf{b,} \emph{Advection:} particles are advected together with the tangential motion of filament at velocity $v_t$.
\textbf{c,} \emph{Deformation:} local particle concentration is changed by the bending and stretching of the curve.
\textbf{d,} \emph{Phase separation:} particles spontaneously partition into dense and dilute phases along the curve.
Feedback among these local processes drives the macroscopic morphological changes illustrated schematically in figure \textbf{a}.}
    \label{fig: 2}
\end{figure*}

Even for a closed elastic curve in the plane, there are many choices to be made in specifying precisely a continuum model that captures the physics of geometry-concentration coupling. Different choices, especially in the density (composition) sector, may be motivated by different microscopic settings in which the concentration field has a different interpretation. Here we are interested in identifying new phenomena rather than quantitative prediction, so we keep these choices as simple as possible, consistent with the physical considerations given in the introduction. 

We consider a closed elastic filament in $\mathbb{R}^2$ coupled to a scalar density field $c$ distributed along its contour; see Fig.~\ref{fig: 2}. The filament undergoes overdamped motion such that the local thermodynamic force density acting at each point on it is balanced by the local viscous drag. The filament is modeled as a one-dimensional elastic curve that can undergo extension, but with a strong enough energetic penalty to enforce near-inextensibility across much of the phase diagram. Its topology is fixed to be that of a simple loop, with turning number or rotation index \cite{do2016differential} of unity, so that the tangent vector rotates by $2\pi$ on circuiting the loop. 

The scalar field variable $c\in[0,1]$ represents a local particle occupation ratio relative to a stretch-dependent maximal capacity: local elongation increases the admissible number of particles, and $c=1$ is always close-packed. This makes sense, {\em e.g.,} for the adsorbed particles in Fig.~\ref{fig: 2} but a different choice might be appropriate in other types of system, such as when $c$ describes the relative proportions of two species in a binary mixture (albeit with any difference disappearing in the inextensible limit).
The particles of density $c$ are assumed to move coherently with any flow along the filament: they are fully
advected by the tangential velocity without slip, and also without exchange in the normal direction via adsorption/desorption. 

The interaction and diffusion of the concentration field involve two natural distance measures: the physical arc length $s$ along the deformed curve, and the material distance along the Lagrangian coordinate $\sigma$. The two coincide in the strictly inextensible limit but differ on a stretched filament. In our near-inextensible regime, we frame the physics in the true arc length $s$; the Lagrangian coordinate $\sigma$ is then used for numerics, where it provides a fixed computational domain without remeshing.\\
Specifically, we take $c(s, t)$ to obey the standard Cahn--Hilliard equation along $s$, with the same form as in an unstretched 1D domain (with no explicit stretch-dependent terms). This is the natural choice for adsorbed particles undergoing 1D liquid-vapour phase separation. An alternative would be to impose Cahn--Hilliard dynamics in the Lagrangian frame, for example a filament composed of two molecular species that swap places with their nearest neighbours at rates that depend on the character of those neighbours and on curvature only, so that these rates are stretch-independent.

More generally, various material parameters in our model are taken to be fully geometry-independent. Such parameters include both $M$, the collective mobility, and $\varepsilon^2$, the square gradient coefficient, of the $c$-field dynamics. (In our units, $\varepsilon$ is then the arc-length width of the interfaces between phase-separated domains.) Likewise the elastic constants governing stretching and curvature are geometry-independent, so that the free energy is quadratic in curvature and stretch (when expressed in physical arc-length coordinates). Another geometry-independent parameter is $\kappa_0$ which controls the energetic coupling between curvature and concentration, such that the local spontaneous curvature is $\kappa_0 c(s)$. Again, this is a specific choice, such that in the absence of adsorbed particles the filament prefers to be straight. However, since the integral of the curvature is fixed by topology (Eq.~\eqref{eq: turning} below), any other baseline for the spontaneous curvature at $c = 0$ would give similar results. The explicit role of all these parameters in the model equations is set out in Secs.~\ref{sec: fieldcurve}, \ref{sec: fieldenergy} below.

For simplicity, we consider only the limiting case of strong phase separation. This means that, in the absence of curvature, the dense (liquid) phase always has $c=1$ and the dilute (vapour) phase $c=0$. For a strictly inextensible filament these can both be rescaled to recover an arbitrary miscibility gap. However, in the extensible case, having $c=0$ in the vapour simplifies our treatment because its density is then dynamically unaffected by stretching.

The above choices are made without further discussion in Secs.~\ref{sec: fieldcurve}, \ref{sec: fieldenergy} where we construct the governing free energy and the dynamical equations for the filament with geometry-concentration coupling. Before doing this, we recall in Sec.~\ref{sec: kinematics}  the fundamental kinematics of an elastic filament from the standpoint of differential geometry.

\subsection{Geometry and kinematics of a closed filament}
\label{sec: kinematics}

\subsubsection{Parametrization and metric.}
Here we follow the formulation of~\cite{nakayama1993motion, gross2025elastic} to model the elastic filament. Let $\Gamma(t)\subset \mathbb{R}^2$ denote a closed material loop, parameterized by a material
label $\sigma\in S^1$ (see Fig.~\ref{fig:curve-mapping}). The spatial embedding is given by a smooth function
$
\mathbf{x}(\sigma,t) : S^1\times[0,T]\to\mathbb{R}^2.
$
The local stretching factor $h(\sigma,t)$ and the metric $g(\sigma,t)$ are
$
h(\sigma,t) = \left| \partial_\sigma \mathbf{x}(\sigma,t)\right|, g(\sigma,t) = h^2(\sigma,t).
$
Let $s$ denote the arc-length coordinate along $\Gamma(t)$. Then the infinitesimal arclength element along the curve is
$
ds = h\,d\sigma$. The total length of the loop is $L(t) = \int_{\Gamma(t)} ds=\int_{S^1} h(t)d\sigma$. Let $\partial_t$ denote the time derivative at fixed material label $\sigma$, while spatial derivatives are taken with respect to physical arclength $s$ via $\partial_s=h^{-1} \partial_\sigma$.
The material coordinate $\sigma\in S^1$ will serve as a fixed computational domain in the numerical
discretization, which is particularly convenient both for pseudo-spectral methods and for optimization methods.
\vspace{-15pt}
\begin{figure}[h!]
    \centering
    \begin{tikzpicture}[scale=1.0*\linewidth/11cm]
    \input{figures/curve_mapping.tex}
    \end{tikzpicture}
    \caption{
Illustration of a closed elastic curve. The Lagrangian (material) coordinate $\sigma \in S^1$ parameterises the reference loop, and the embedding $\mathbf{x}(\sigma, t)$ maps it to the physical curve $\Gamma \subset \mathbb{R}^2$, on which the Eulerian arclength coordinate $s$ is measured.
}
    \label{fig:curve-mapping}
\end{figure}
\vspace{-30pt}
\subsubsection{Frenet frame and geometric evolution.}
The unit tangent and normal vectors are denoted by $\mathbf{T}=\partial_s \mathbf{x}$ and $\mathbf{N}$, respectively. The curvature $\kappa(s, t)$ is defined by the Frenet-Serret relation
$
\partial_s \mathbf{T}=\kappa \mathbf{N}.
$
The velocity of material points along the evolving curve can be decomposed in the Frenet frame as
\begin{equation}
    \partial_t \Gamma=v_t \mathbf{T}+v_n \mathbf{N},
\end{equation}
where $v_t(s, t)$ and $v_n(s, t)$ denote the tangential and normal velocity components. 
In two dimensions, the kinematic evolution equations for the metric $g$ and curvature $\kappa$ are 
\begin{equation}\label{eq: gtkt}
    \begin{split}
        \partial_t g &= 2g\left( {{\partial _s}{v_t} - \kappa {v_n}} \right)\,,\\
    \partial_t \kappa &=   (\partial _s^2 + \kappa^2 ){v_n} + ({\partial _s}\kappa ){v_t}\,.
    \end{split}
\end{equation}
In what follows, it will also be convenient to use the equation 
\begin{equation} \label{eq: ht}
\partial_t h =h\left( {{\partial _s}{v_t} - \kappa {v_n}} \right)\,.
\end{equation}
A detailed derivation, and more general discussion, of these results are given in Appendix \eqref{secapp: variation}.

Note that a more general framework of filament and frames is discussed in the works~\cite{kikuchi2024rare, nemeth2024geometric, epstein1998geometrical} by using the concepts of Cosserat media and moving frames. 

\subsubsection{Reconstruction and closure constraints.}
\label{sec: recon}

To reconstruct the embedding $\mathbf{x}$ from its metric and curvature, we introduce the tangent angle $\theta(s,t)$ such that
$
\mathbf{T}(s,t) = (\cos\theta(s,t),\,\sin\theta(s,t)).
$
From $\partial_s\mathbf{T}=\kappa\mathbf{N}$ it follows that
$
\partial_s\theta = \kappa.
$
Given an initial orientation $\theta(0,t) = \theta_0(t)$, we obtain
\begin{equation}\label{eq: theta}
    \theta(s,t) = \theta_0(t) + \int_0^s \kappa(s',t)\,ds'.
\end{equation}
The loop position is then reconstructed by integrating the tangent:
\begin{equation}
    \mathbf{x}(s,t)
    = \mathbf{x}(0,t)
    + \int_0^s \mathbf{T}(s',t)\,ds'.
    \label{eq:reconstruction}
\end{equation}
Closure of the loop requires the endpoint coincides with the starting point, leading to the geometric {\bf closure constraint} involving the double integral of the curvature:
\begin{equation}
    \int_{\Gamma(t)}
    \begin{pmatrix}
    \cos\theta(s,t)\\[2pt]
    \sin\theta(s,t)
    \end{pmatrix}
    ds
    =
    \begin{pmatrix}
    0\\[2pt]
    0
    \end{pmatrix}\,.
     \label{eq:closure}
\end{equation}
In addition, the tangent angle must return to its original direction after one full loop, giving the topological {\bf turning number constraint}
\begin{equation}
    \int_{\Gamma} \kappa(s,t)\,ds = 2\pi m
    \label{eq: turning}
\end{equation}
with turning number $m\in\mathbb{Z}$ and $m=1$ for a simple loop.

In numerical implementations, particularly when using pseudo-spectral methods (as discussed in Section~\ref{sec: pseudo} below), it is convenient to work on a fixed computational domain. Since the arc-length coordinate $s$ varies in time as the loop stretches or contracts (albeit this effect is modest in most parameter regimes of interest), we express all geometric quantities in the material coordinate $\sigma$, such that the integration domain becomes fixed. 

\subsection{Scalar field on a curve}
\label{sec: fieldcurve}

We now consider the order parameter field $c(s)$ representing a local occupancy
fraction of particles (such as the adsorbed particles in Fig.~\ref{fig: 2}). This takes values in $[0,1]$, where $c=0$ denotes an empty filament and $c=1$ denotes
maximal coverage ({\em e.g.}, close packing).

\subsubsection{Definition and normalization.}
To define $c$ and determine its kinematics, we first introduce a cumulative particle number $N(\sigma, t)$ on the unstretched material coordinate $\sigma$. 
We denote by $\rho_E(s,t)$ the particle number per unit physical (Eulerian) arc length and by $\rho_L(\sigma,t)$ the corresponding density in the (Lagrangian) reference frame coordinate. Then we have
\begin{equation} \label{eq: rho_L}
    \rho_L(\sigma, t)=\frac{d N}{d \sigma}, \quad \rho_E(s, t)=\frac{d N}{d s}=\frac{\rho_L}{h} .
\end{equation}
We assume strong phase separation, whereby the maximal coverage is attained in the dense phase of an undeformed (unstretched and flat piece) of filament at coexistence. Calling this density $\rho_0$, the concentration or occupational order parameter is constructed as
\begin{equation}
    c(s,t)= \frac{\rho_E(s,t)}{\rho_0}  = \frac{\rho_L(\sigma,t)}{\rho_\text{max}(\sigma,t)} = \frac{\rho_L(\sigma,t)}{h(\sigma,t) \rho_0} 
    \label{eq:def_c}
\end{equation}
where we have defined the $\sigma$-space maximal density as 
\begin{equation}
    \rho_\text{max}(\sigma,t)=h(\sigma,t) \rho_0\,.
\end{equation}
This increases on stretching the filament such that the maximal density per unit \emph{physical} arc-length remains $\rho_0$.

\subsubsection{Conservation law and density kinematics.}
For the entire loop, the global total concentration is conserved, having a time-independent total amount $C_0$ such that 
\begin{equation}
  C_0:=\int_{\Gamma(t)} c(s,t) d s, \quad \forall t \,.
\end{equation}
For comparison across different parameter sets, we introduce the dimensionless conserved total coverage
\begin{equation}\label{eq:c_total_nd}
C := \frac{\int_{S^1} c(s,t)h(t)\,d \sigma}{\int_{S^1} h_0\,d\sigma}, \quad \forall t \,.
\end{equation}

Adsorbed particles are assumed to be bound to the elastic filament and carried along by its motion in the ambient space.  On top of this they move along the tangent to the filament under the phase separation dynamics. The velocity of a particle at position $x^p$ is decomposed in the Frenet frame as
$
\partial_t x^p = v_t^p \mathbf{T} + v_n^p \mathbf{N}.
$
There is no desorption, so that particles share the normal velocity of the filament. Distinguishing the motion of the material points of the filament from the self-dynamics of the density field, we then have
\begin{equation}
    v_n^p = v_n\,; \quad  v_t^p = v_t + v_{\mathrm{rel}}
    \label{eq:vn_particle}
\end{equation}
where $v_{\mathrm{rel}}$ is a relative slip generated by the intrinsic dynamics of the field, specifically, collective diffusion driven by chemical potential gradients. 

This means that, in the absence of $v_{\mathrm{rel}}$, no particle contributing to the $c$ field passes a physical marker placed on any material segment of the underlying filament due to the filament deformation
\begin{equation}
\delta \mathbf{x}=\xi \mathbf{T}+\eta \mathbf{N}\,,\label{eq: shapech}
\end{equation}
where $\xi$ and $\eta$ are the tangential and normal deformation in the local frame. This indicates that, during deformation of the filament, the material is attached to the interface, and changes in the density in physical space are determined by the physical space distance between two material points. 

Delineating an interval of Eulerian arc-length coordinates by two such markers as $\omega(t)=\mathbf{x}([\sigma_1, \sigma_2],t) \subset \Gamma(t)$,
conservation of the density $\rho_E$ then implies
\begin{equation}
\delta \int_{\omega(t)} \rho_E d s=\rho_0 \cdot \delta \int_{\sigma_1}^{\sigma_2} c(\sigma) h d \sigma = 0\,,
\end{equation}
which provides the {\bf advection relation},
\begin{equation}
    \delta  c(\sigma) = -c\left(\partial_s \xi-\kappa \eta\right) \,. \label{eq: exp}
    \end{equation}
Note that $\xi$ and $\eta$ are infinitesimals, as defined via \eqref{eq: shapech}. For a fuller derivation see Appendix~\ref{secapp: variation}. 

Eq.~\eqref{eq: exp} describes the change in $c$ in the material frame caused solely by expansion of the filament. The only remaining change is due to transport by the relative velocity term $v_\text{rel}$.  In Lagrangian coordinates, this transport is governed by the continuity equation (there being no sources or sinks of particles)
\begin{equation}
    \partial_t \rho_L + \partial_\sigma J_L =0\,,
    \label{eq:lagrangian_conservation}
\end{equation}
where $J_L(\sigma)$ is the Lagrangian particle-number flux, that is, the number of particles passing a material point in unit time. Combining this with the advection term, we have
\begin{equation} \label{eq: c_t}
    \partial_t c(s)  =\rho_0^{-1} \partial_t\left(\rho_L /h\right)
=-c\left(\partial_s v_t-\kappa v_n\right) - \partial_s J \,,
\end{equation}
where we introduce $J = \rho_0^{-1} J_L=cv_{\mathrm{rel}}$ as the normalized coverage flux; this can be viewed as the diffusive flux of the density field $c(s)$ along the filament's physical arc-length coordinate $s$. 

In the next subsection, this $J$ will be specified by a gradient-flow relation involving the chemical potential, chosen to describe the physics of phase separation. 

\subsection{Energy and Gradient flow}
\label{sec: fieldenergy}
As emphasized above, the evolution of the concentration field $c(s)$ takes place along the physical arc-length coordinate $s$. Denoting the curvature of the filament by $\kappa(s)$, we first introduce a one-dimensional Helfrich-type bending energy with dimensional rigidity parameter $\beta_{\kappa}$,  in which the local spontaneous curvature with particle coverage $c$ is $\kappa_{0} c$:
\begin{equation}\label{eq: bend}
    E_{\text {bending}}[\mathbf{x},c]=\frac{\beta_{\kappa}}{2} \int_{\Gamma}\left(\kappa-\kappa_0c\right)^2 d s\,.
\end{equation}
Thus we address the case where a section of exposed filament with no adsorbed particles prefers to be straight, while a fully covered one has preferred curvature $\kappa_0$. The sign is chosen so that regions of large $c$ curve inwards more strongly (as seen in Fig.~\ref{fig: metastables} below).

To model stretchability, we employ a local linear elasticity with energy
\begin{equation} \label{eq: E_metric}
E_{\mathrm{metric}}[\mathbf{x}]=\frac{\beta}{2} \int_{S^1}(h(\sigma)-h_0(\sigma))^2 d \sigma\,,
\end{equation}
where as usual the local extensional strain measures the difference in arc-length separation between material points and the same separation in the undeformed reference state. Here $h_0(\sigma)$ is the distance metric between material points in the unstretched curve. 
The free energy functional for the concentration field $c(s)$ is of Cahn-Hilliard type, comprising a local (dimensionless) double-well potential $W(c)=\tfrac{1}{4}c^2(1-c)^2$ and a square-gradient correction responsible for interfacial tension between phases:
\begin{equation}\label{eq: cahn}
    E_{\mathrm{field}}[c]
    = \alpha \int_{\Gamma}
    W(c)
    + \tfrac{\varepsilon^2}{2}|\partial_s c|^2 ds\,,
\end{equation}
where $\alpha$ has dimensions of energy per unit length. Note however that the energetics of the $c$ field additionally depend on curvature via the coupling in \eqref{eq: bend}. Only if $E_{\text{bend}}$ is zero (which is not generically possible for a closed filament) is \eqref{eq: cahn} the sole driver of the $c$ field evolution. We note that the polynomial double-well $W(c)$ adopted here is the standard Landau expansion of a mixing free energy, so that both the energetic and entropic contributions of the underlying phase-separating system are already encoded in $W(c)$. The findings of equilibria reported in this paper depend mainly on the existence of two coexisting phases and should carry over qualitatively to other free energies that support binodal coexistence, including ones that distinguish entropic and interaction terms (such as Braggs-Williams or Flory-Huggins free energies, or more recent microscopic models\cite{Frey2023mechanochemical, PhysRevE.111.044405}).

Combining the above three energies, we construct the (dimensionful) free energy
functional \begin{equation}
E[\mathbf{x},c]
=
E_{\rm bending} + E_{\rm metric} + E_{\rm field} \,.
\label{eq: dim_energy_main}
\end{equation}
To obtain the force density on the filament, we consider a filament deformation \eqref{eq: shapech} and define the variation by
\begin{equation}
\delta E
=
\int_\Gamma \Big(F_{t}\,\xi +F_{n}\,\eta\Big)\,ds
+\int_\Gamma \mu\,\delta c\,ds.
\label{eq:shape_deriv_def}
\end{equation}
Here $F_{t}=\frac{\delta E}{\delta \mathbf{x}} \cdot \mathbf{T}$ and $F_{n}=\frac{\delta E}{\delta \mathbf{x}} \cdot \mathbf{N}$ are the force densities 
in the tangential and normal directions, and $\mu=\frac{\delta E}{\delta c}$ is the chemical potential. Their explicit forms are obtained by the geometric variation identities (Appendix~\ref{secapp: variation}) together with~\eqref{eq: dim_energy_main}.

We assume an overdamped limit with a local damping coefficient $\zeta$ per unit arclength, so that the filament
velocity is proportional to the local thermodynamic force density:
\begin{equation}
\zeta\,v_t = -\,F_{t},\qquad
\zeta\,v_n = -\,F_{n}.
\label{eq:overdamped_dim_main}
\end{equation}

We now prescribe Cahn--Hilliard dynamics \cite{10.1063/1.1744102} along the physical arclength
coordinate $s$ with a constant mobility $M$:
\begin{equation}
J = -M\,\partial_s \mu,
\label{eq:flux_dim_main}
\end{equation}
where $J$ denotes the normalized coverage flux in equation \eqref{eq: c_t} and the corresponding particle number flux is $\rho_0 J$.

\subsection{Nondimensional equations} \label{sec: nondim}

As shown in Appendix~\ref{secapp: nondimension}, we may nondimensionalize the free energy $E$
\eqref{eq: dim_energy_main} and the overdamped dynamics by choosing the characteristic scales
\begin{equation}
L_0 := \frac{1}{2\pi}\int_{S^1} h_0(\sigma)\,d\sigma,\quad
E_0 := \frac{\beta_{\kappa}}{L_0},\quad
T_0 := \frac{\zeta L_0^4}{\beta_{\kappa}},
\label{eq:scales_main}
\end{equation}
where $2\pi L_0$ is the unstretched filament length, $E_0$ is the bending-energy scale, and $T_0$ is the damping time scale. In these units, the energetic parameters become (with tilde symbols for nondimensional quantities) 
\begin{equation}
\tilde \alpha \,:=\, \frac{\alpha L_0^2}{\beta_{\kappa}}, \;\;
\tilde \beta \,:=\, \frac{\beta L_0^3}{\beta_{\kappa}}, \;\;
\tilde \varepsilon \,:=\, \frac{\varepsilon}{L_0}, \;\;
\tilde \kappa_0 \,:=\, \kappa_{0} L_0 \,, \label{eq:nodimgroups}
\end{equation}
and the dynamical parameters become
\begin{equation}
    \tilde M \,:=\, \zeta M, \quad
    \tilde \zeta \,:=\, 1.
\end{equation}
It is important to notice that increasing the reference length scale $L_0$ increases the stretching stiffness $\beta \sim L_0^3$ faster than the other dimensionless groups in \eqref{eq:nodimgroups}. Hence the dynamics approaches the near-inextensible regime when  $L_0 $ becomes large.

Including the global concentration parameter $C$ which is already dimensionless  \eqref{eq:c_total_nd} we arrive at the nondimensional parameter set  $\{{\alpha},\,{\beta},\,{\varepsilon},\,{\kappa}_0,\, m,\,{C}\}$, where we drop the tilde symbols so that in the rest of the paper all parameters are nondimensional unless specified otherwise.

\begin{widetext}
In Lagrangian coordinates $\sigma$ and with nondimensionalized field variables $h, \kappa, c$, the free energy then takes the non-dimensional form:
\begin{equation}
\mathcal{E}[\mathbf{x},c]
=\frac{1}{2}\int_{S^1} (\kappa-\kappa_0 c)^2h\,d\sigma
+\frac{\beta}{2}\int_{S^1}(h-h_0)^2\,d\sigma
+\alpha\int_{S^1}\Big[\,W(c)+\frac{\varepsilon^2}{2}\,(\frac1h\partial_\sigma c)^2\,\Big]h \,d \sigma\,,
\label{eq:energy_full_nd}
\end{equation}
from which the nondimensional chemical potential $\mu$ is found as:
\begin{equation} \label{eq:chemical_potential_nd}
\mu=\frac{\delta \mathcal{E}}{\delta c}=\alpha\mu_{\rm field}-\kappa_0\left(\kappa-\kappa_0 c\right); \qquad \mu_{\rm field}:=W'(c)-\varepsilon^2 \partial_s^2c; \qquad f:=W(c)+\frac{\varepsilon^2}{2}|\partial_s c|^2\,.
\end{equation}
Here, the derivative operator is $\partial_s = h^{-1}\partial_\sigma$ and the double-well potential energy is $W(c)=\tfrac{1}{4}c^2(1-c)^2$ as defined previously; $f$ is the corresponding free energy density including the square gradient part. 

The resulting gradient flow dynamics is then a coupled nonlinear PDE system combining \eqref{eq: gtkt}, \eqref{eq: ht} and \eqref{eq: c_t}:
\begin{equation}
    \begin{split}
        \partial_t h &= h\left( \partial_s v_t - \kappa v_n \right), \\[3pt]
        \partial_t \kappa &= \left( \partial_s^2 + \kappa^2 \right) v_n + (\partial_s \kappa) v_t, \\[3pt]
        \partial_t c &=- c(\partial_s v_t - \kappa v_n) + M\partial_s^2 \mu\,, 
    \end{split}
    \label{eq:evolution_gkc}
\end{equation}
where the velocities obey the overdamped dynamics:
\begin{equation}
    \begin{split}
v_t &
= \kappa_0c\,\partial_s\!\big(\kappa-\kappa_0 c\big)
+\partial_s\big(\beta(h-h_0) - \alpha\mathcal P\big),
\\[4pt]
v_n &
= -\partial_s^2 \big(\kappa-\kappa_0 c\big)
-\frac12\,\kappa\big(\kappa-\kappa_0 c\big)^2
+\kappa \big(\beta(h-h_0) - \alpha\mathcal P\big)\,.
\end{split}
    \label{eq:vt_vn}
\end{equation}
Here the chemical stress $\mathcal{P}$ derives from the compositional order parameter via the Cahn-Hilliard free energy \eqref{eq:energy_full_nd}  \cite{kendon2001inertial}: 
\begin{equation} \label{eq: P_stress}
\mathcal P:=c\,\mu_{\rm field}-f +\varepsilon^2 |\partial_sc|^2 \,.
\end{equation}
\end{widetext}
For numerical convenience, we adopt below a uniform reference metric $h_0=1$, so that the unstretched filament length is $\int_{S 1} h_0 d \sigma=2 \pi$.
Equations \eqref{eq:evolution_gkc} and \eqref{eq:vt_vn} form a coupled three-fields, fourth-order nonlinear PDE system of Willmore-type gradient flow\cite{willmore1993riemannian} and Cahn-Hilliard flow, which in our reduced units, fully specifies the evolution of the closed filament and its internal concentration field, subject to the global constraints.

As previously discussed, the global constraints are: closure of the curved filament, \eqref{eq:closure}; conservation of the turning number \eqref{eq: turning}, and conservation of global concentration \eqref{eq:c_total_nd}. The last of these is automatically sustained via the Cahn-Hilliard dynamics, whereas the first two can easily be written in dimensionless form for use with the dynamical equations (\ref{eq:evolution_gkc}, \ref{eq:vt_vn}) and/or algorithms that directly minimize \eqref{eq:energy_full_nd}.

Given the geometric fields $(h,\kappa)$, we reconstruct the embedding by first computing $\theta$ from
\eqref{eq: theta}, and then integrating the tangent as in \eqref{eq:reconstruction} but now in $\sigma$ coordinates:
\begin{equation}
\mathbf{x}(\sigma,t)
=
\mathbf{x}(0,t)
+\int_{0}^{\sigma}
\begin{pmatrix}
\cos\theta(\sigma',t)\\[2pt]
\sin\theta(\sigma',t)
\end{pmatrix}
h(\sigma',t)\,d\sigma' .
\label{eq:reconstruction_sigma}
\end{equation}
This procedure creates the physical space curve that defines the locus of the elastic filament in the plane.

\section{ANALYSIS} \label{sec: analysis}

On a closed elastic filament, phase separation is no longer driven solely by the standard Cahn-Hilliard equation as in a flat 1D space with periodic boundary conditions; instead, it exhibits a nontrivial coupling to the local curvature which is further affected by the constraints on global geometry and topology. The closure and turning number requirements of Section~\ref{sec: recon} impose integral constraints on the curvature distribution. Specifically, any infinitesimal shift of a phase boundary will alter the lengths of the adjacent domains and potentially violate the curvature integral constraints unless some part of the system deviates from its locally preferred curvature. To maintain closure, the curvature field must readjust across the entire loop. Thus the curvature distribution influences interface energetics, and interface placement influences the curvature distribution. This yields a strongly coupled and nonlocal variational system. 

Such feedback does not exist for a circular geometry that is fixed and (metrically) flat: here the motion of interfaces between phase-separated domains does not induce any global geometric response. The dynamics and free energy are then just those of phase separation in a 1D domain with periodic boundary conditions. It is known in this case that any configuration with more than one domain of each type is unstable because there is a weak attraction between interfaces (albeit exponentially small at large separations). Any multi-domain state is connected by a path of strictly decreasing $\mathcal{E}$ to one with just one domain of each type -- here, one liquid and one vapour domain. Hence there are no metastable states.
In contrast, phase separation on a closed elastic filament involves a global competition between the scalar concentration field and the loop's geometry. The admissible configurations, metastable states, and energy landscape are fundamentally altered, admitting a class of morphology selection mechanisms that arise only for a closed curve with geometry-concentration coupling.

\subsection{Sharp-interface limit}
Analytic progress in calculating the free energy $\mathcal{E}$ for phase separated states can be made in the sharp-interface limit, that is, $\varepsilon/L_0\to 0$. (We also take the inextensible limit, $\beta\to\infty$, in this Section.)

First, consider the case of an open filament. Here the bending energy can be minimized to zero with any pattern of $c$-field domains by setting $\kappa = \kappa_0 c$ everywhere -- giving straight segments in the low-density (vapour) regions and circular arcs of curvature $\kappa_0$ in the dense (liquid) ones. Then for a domain pattern with $N$ interfaces the energy is simply $\gamma N$ where $N$ is the number of domain walls (interfaces) and $\gamma$ is the Cahn-Hilliard interfacial tension. This approximation neglects the exponentially small interaction energy between interfaces; the phase-separated states are then neutrally stable with respect to any interfacial displacements that conserve the global concentration $C$, such as translating one chosen liquid-phase domain (and its two interfaces) rigidly along the contour $s$. 

On a closed loop, however, the constraints of closure and turning number drastically restrict this admissible set of (almost) degenerate free-energy minima.
Further details are provided in Appendix \ref{secapp: sharp}. The results are formally best understood by first introducing Lagrange multipliers $\lambda$ for the turning number in the presence of a fixed concentration pattern $c(s)$ but ignoring closure, and seeking minima of: 
\begin{equation}
\begin{split}
        \mathcal{L}[\kappa, c]=&\frac{1}{2} \int_{\Gamma}\left(\kappa-\kappa_0 c\right)^2 d s+\gamma N 
        -\lambda\left(\int_{\Gamma} \kappa d s - 2 \pi m \right)
\end{split}
\end{equation}
with $\lambda$ chosen to enforce
$
\int_{\Gamma} \kappa d s=L\bar\kappa
=2 \pi m$ for turning number $m$  (see \eqref{eq: turning} above).

Proposition \ref{prop: optimal_kappa} in Appendix \ref{secapp: sharp} shows that this full variation with respect to $\kappa$ gives $\lambda = \bar \kappa - \kappa_0\bar c $  with $\bar c \equiv C_0/L$, 
and hence the optimal curvature configuration $\kappa^*(s)$ obeying 
\begin{equation}\label{eq: optimal_kappa}
    \kappa^*(s)=\bar\kappa +\kappa_0(c(s)-\bar c)\,.
\end{equation}
Thus, when the closure constraint is ignored, but with turning number $m$  imposed, all spatial domain patterns with $N$ sharp interfaces and the required global concentration $C_0$ are degenerate. The free energy minimiser for such a filament is composed of alternating circular arcs of curvature 
\begin{equation} \label{eq: kappa_plus}
    \kappa_{+}= \kappa_0 + \lambda\;,  \qquad  \kappa_{-}=  \lambda\;.
\end{equation}
The energy in this case is 
\begin{equation} \label{eq: energy_sharp}
    \mathcal{E}_{\text {sharp }}=\frac{L}{2}\left(\bar{\kappa}-\kappa_0 \bar{c}\right)^2+\gamma N,
\end{equation}
which only depends on the globally constrained values of $m,\bar{c}$, and on $N$. This bending energy $\tfrac{L}{2} \lambda^2$ should be regarded as the necessary bending energy that must be paid in the system with that amount of concentration to satisfy the turning constraint. This result directly generalizes the case of  an open filament with completely unconstrained curvature, discussed in the previous paragraph.

If we finally return to the fully closed filament problem, this obeys the closure constraint \eqref{eq:closure}, alongside the turning number constraint and $c$ conservation as just imposed. Any such closed filament must also have even $N$.
Importantly however,  the turning number relation \eqref{eq: optimal_kappa} and closure condition \eqref{eq:closure} cannot in general both be satisfied for $N=2$: fine-tuning of other parameter values is needed to find a solution. (See Appendix~\ref{secapp: sharp} for further discussion.) However, they can generically be satisfied for $N\geq 4$ so that, in the sharp interface limit, it is never necessary to have more than two domains of each type (four domain walls in total) in order to minimize the global free energy. 

\begin{figure*}
  \centering
    \includegraphics[width=.95\linewidth]{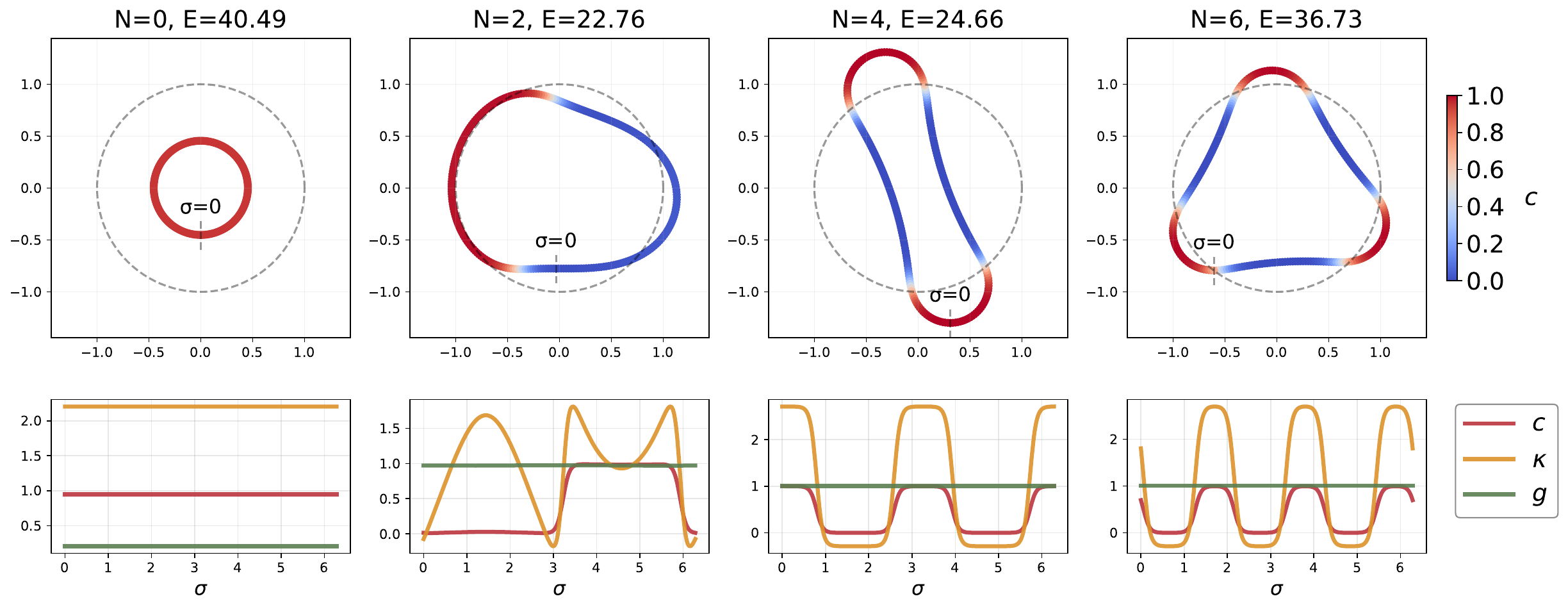}
    \caption{
    Free energy minimisers for the case $\alpha=1024, \beta=20, \epsilon=0.05, \kappa_0=3, C=0.43$, classified left to right by interface number: circle $(N=0)$, acorn $(N=2)$, peanut $(N=4)$, polygon $(N \geq 6)$. Top row: the reconstructed real-space embeddings, colour coded by $c$; the dashed grey circle indicates the curve at rest in the absence of the coupled concentration field, and the small marker labels the position $\sigma=0$. Bottom row: the corresponding profiles of the concentration $c$ (red), curvature $\kappa$ (orange), and metric $g$ (green) along the Lagrangian coordinate $\sigma \in[0,2 \pi$). The $N=0,4,6$ cases satisfy the optimal relation ~\eqref{eq: optimal_kappa}, with $\kappa$ essentially piecewise constant on alternating arcs ($\kappa$ for $N=0$; the two values $\kappa_{ \pm}$for $N \geq 4$ ). The $N=2$ (acorn) case is qualitatively different: the closure constraint ~\eqref{eq:closure} cannot in general be met by a single pair of circular arcs of constant curvatures $\kappa_{ \pm}$, so the system pays additional bending energy by deforming each domain away from a circular arc. This is visible in the top-row for the case $N=2$, in contrast to the nearly flat $\kappa$ plateaus seen for $N=4,6$.
}
     \label{fig: metastables}
\end{figure*}

Moreover, from \eqref{eq: energy_sharp}, we identify an important value $C=C^*$ that satisfies $\lambda = \bar{\kappa}-\kappa_0 \bar{c}=0$:
\begin{equation}
C^* = \frac{1}{\kappa_0}\,.    
\label{eq: Cstar}
\end{equation}
At this global density the high-density phase can in principle accommodate the full turning budget without curvature mismatch and the bending energy becomes minimal. At this value, the filament comprises pieces with curvature $\kappa_0$ and $0$ obeying the equation \eqref{eq: kappa_plus}. The corresponding minimizer is typically capsule/rod-like, with two nearly straight segments joined by semicircular endcaps.

\subsection{Minimizers}

The above findings for the sharp interface limit motivate the following classification for minimizers (either stable or metastable) of $\mathcal{E}$ in terms of their interface number, $N$. For any particular parameters there is generically only one stable $N$ value but, quite unlike the case of phase separation on a fixed circle, there can also be metastable ones.
\begin{itemize}
    \item \textbf{$N=0$ (uniform shape)}. The loop minimizes the interfacial energy but must accommodate a global mismatch between the average curvature $\bar{\kappa}=2 \pi m / L$ and the preferred curvature $\kappa_0 \bar{c}$.
Inextensible loops can satisfy closure only by maintaining this mismatch. (Note however that extensible ones can also relieve it through changes of total length, effectively diluting or condensing the composition so as to move $\kappa_0\bar c$ closer to the required $\bar \kappa$.) 
\item \textbf{$N=2$ (acorn shape)}. The system introduces one pair of interfaces to separate the $c$ field into two coexisting domains. Although this reduces curvature-composition frustration, the closure constraint cannot generally be satisfied by two circular arcs of fixed curvature, requiring non-circular-arc deformations that increase bending energy.
\item  \textbf{$N=4$ (peanut shape)}. The configuration provides the minimal even partition that permits geometric closure by alternating four arcs of curvatures $\kappa_{ \pm}$. This $N=4$ solution family has dimension zero, hence the symmetric arrangement with two pairs of identical domains is the unique minimizer of this type, in the sharp interface limit. While often biconcave (a peanut), cases where the two domain types have the same sign of curvature creating a convex loop are also often found as minimizers (for examples of this see Section~\ref{sec:results} below).
\item \textbf{$N \geq 6$ (polygon shape)}. These sharp-interface configurations form a continuous family of minimizers with dimension $N-4$ under the three global constraints of closure, turning, and total composition. Such a minimizer is not a global miminum when we consider these three constraints, since it pays for two or more additional interfaces without any additional reduction in the elastic energy of the filament. Note however that any  further global constraints (such as a fixed enclosed area) would in general not be satisfied by the $N=4$ solution, moving the global energy minimizer to $N\ge 6$.
\end{itemize}

Examples of each type are shown in Fig.~\ref{fig: metastables}. The respective (nondimensionalized) energetic contributions for each of these examples is in Table \ref{table:metastables}.
\begin{table}[hbp]
  \centering
  \begin{tabular}{lccccc}
    \hline
    Case & $E$ & $E_{\text{bending}}$& $E_{\text{well}}$ & $E_{\text{interface}}$& $E_{\text{metric}}$\\
    \hline
    N=0 & $4.1\times10^{1}$ & $1.2\times10^{0}$ & $1.8\times10^{0}$ & $0.0\times10^{0}$ & $3.8\times10^{1}$ \\
    N=2 & $2.3\times10^{1}$ & $1.0\times10^{1}$ & $6.6\times10^{0}$ & $5.8\times10^{0}$ & $2.6\times10^{-2}$ \\
    N=4 & $2.5\times10^{1}$ & $5.2\times10^{-1}$ & $1.2\times10^{1}$ & $1.2\times10^{1}$ & $5.4\times10^{-4}$ \\
    N=6 & $3.7\times10^{1}$ & $5.2\times10^{-1}$ & $1.8\times10^{1}$ & $1.8\times10^{1}$ & $5.4\times10^{-4}$ \\
    \hline
  \end{tabular}
  \caption{Free energy decomposition of the states in Fig.~\ref{fig: metastables}. Alongside the usual bending and stretching terms, we introduce $E_\text{well} = \alpha \int_{\Gamma} W(c) \, ds$ as the bulk double-well free energy contributions of the two coexisting phases (which would be zero for a fixed circular geometry), and $E_\text{interface} = \alpha \int_{\Gamma}\frac{\varepsilon^2}{2}\left|\partial_s c\right|^2 ds$. The latter equates to $N\gamma$ in the sharp interface limit. }
  \label{table:metastables}
\end{table}

\begin{figure*}[t]
    \centering

    \begin{subfigure}[t]{0.49\linewidth}
        \centering
        \includegraphics[width=\linewidth]{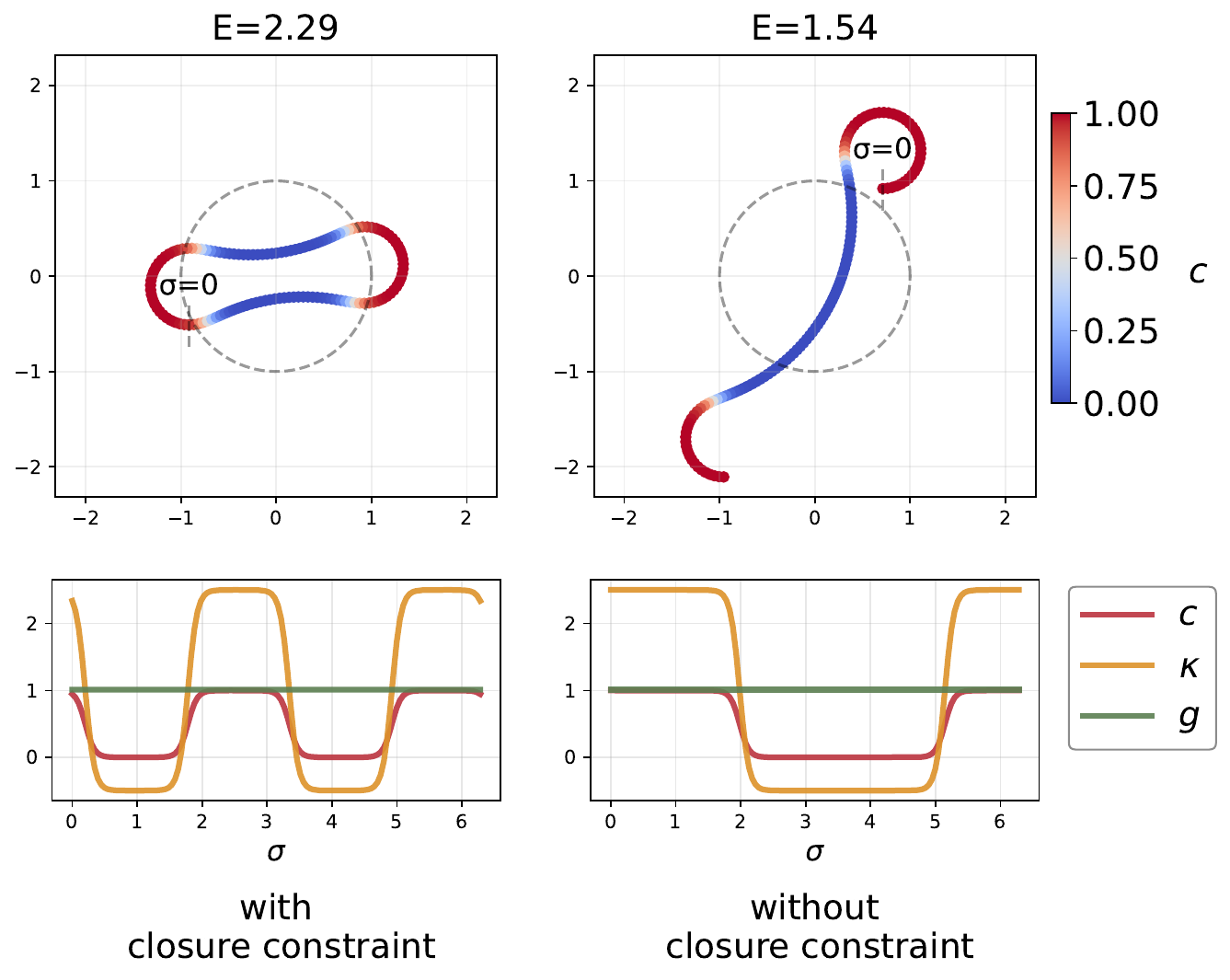}
        \caption{$\alpha=64$, $C=0.5$. }
        \label{fig:closure_compare_a}
    \end{subfigure}
    \hfill
    \begin{subfigure}[t]{0.49\linewidth}
        \centering
        \includegraphics[width=\linewidth]{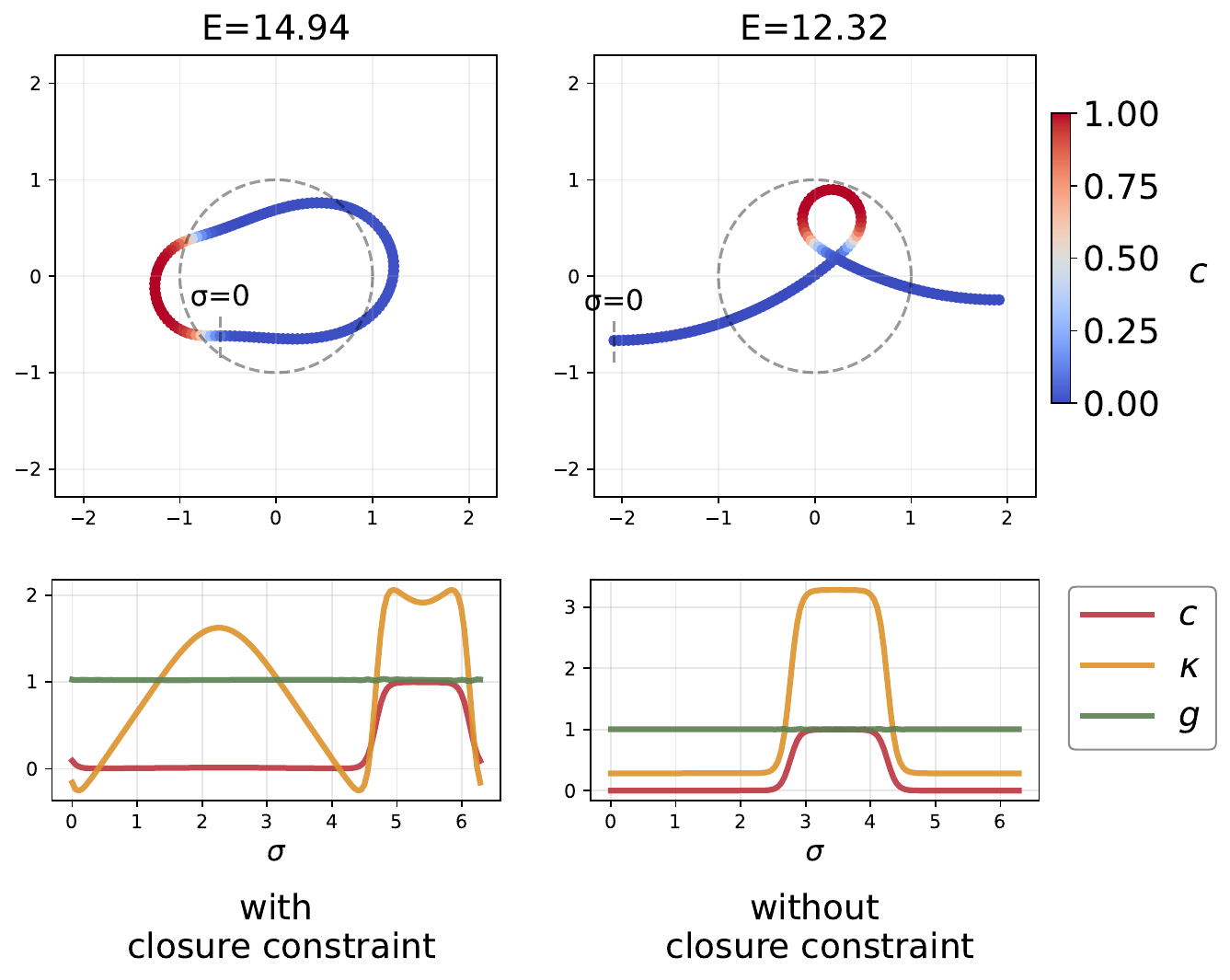}
        \caption{$\alpha=1024$, $C=0.24$.}
        \label{fig:closure_compare_b}
    \end{subfigure}

    \caption{
        Effect of the closure constraint. Upper panels: real-space shapes with concentration field in colour scale. Lower panels: plots of concentration, curvature and metric. For $\alpha=64$ and $C=0.5$ (left), imposing the closure constraint yields a closed global minimum (a `peanut' curve) with $N=4$, whereas with only the turning number constraint~\eqref{eq: turning} and periodic boundary conditions the optimum has $N=2$. The global concentration is $C=0.5$ so that the two phases have the same total arc length. For $\alpha=1024$ and $C=0.24$ (right), the global optimum remains in the $N=2$ (acorn) category; removing the closure constraint eliminates the need to trade bending energy for closure and therefore lowers the optimal energy. Other parameters are $\kappa_0=3, \varepsilon=0.05, \beta=20$ as in Fig.~\ref{fig: metastables}; for more details see Appendix \ref{secapp: numerical}.}
    \label{fig:closure_compare_combined}
\end{figure*}

When the interface width is finite, the degeneracy for $N\ge 6$ configurations is lifted, first because of the weak attraction between interfaces, and secondly because the smooth concentration profiles that now connect neighbouring domains introduce new ways to couple composition and geometry locally. This leads to distinct metastable morphologies that are absent in the sharp-interface limit. An analytical study is challenging, but numerical methods are used to study some of these cases in Section \ref{sec: numerical} below. 

\subsection{Closure constraint}

Before considering our examination of the free energy minimizers and of the dynamics that leads to them, we remark briefly on the issues introduced by the closure constraint  \eqref{eq:closure}. This constraint is essential, as we have seen, when formulating the problem using intrinsic geometry. While much of the previous literature emphasises topological constraints \cite{elliott2021domain}, intrinsic discretisations optimise using quantities defined along the material coordinate $\sigma$ ({\em e.g.,} compositional and curvature fields) that do not, by themselves, enforce geometric closure of the embedded curve in $\mathbb{R}^2$. In particular, periodic boundary conditions on the curvature field do not exclude non-closed embeddings.

The closure constraint is what creates trade-offs between the different energies in the system \eqref{eq:energy_full_nd}. Without it, the optimum is always a single pair of domains ($N=2$); there is no reason ever to increase $N$ because there is no frustration of the curvature energy. (Likewise, without closure, the stretching energy does not trade off nontrivially with the other energy terms.) Second, among $N=2$ configurations (acorn shapes), the closure constraint changes the minimum energy because one cannot generically construct a closed curve from circular arcs with the preferred curvatures $\kappa_\pm$. 
These remarks, which become exact in the sharp interface limit, are illustrated with simulation data (found using the methods introduced in the next Section), for moderately sharp interfaces, in Fig.~\ref{fig:closure_compare_combined}.

\section{NUMERICAL METHODS} \label{sec: numerical}

This Section describes the numerical methods used to study the coupled curvature--composition problem on a closed loop. We work in a fixed material coordinate $\sigma\in[0,2\pi)$, representing the intrinsic fields $h(\sigma),\kappa(\sigma),c(\sigma)$ on a stationary periodic domain. This avoids remeshing under deformation and provides a consistent setting for geometric evolution and field advection. Line integrals are evaluated as $\int_{\Gamma}(\cdot)\,ds=\int_0^{2\pi}(\cdot)\,h(\sigma)\,d\sigma$, which we use to compute energies and constraint residuals.

Within this representation, we (i) compute metastable equilibria via direct constrained minimization of \eqref{eq:energy_full_nd} under mass, turning number, and closure constraints, and (ii) integrate the gradient-flow PDEs using pseudo-spectral discretization and IMEX time stepping to obtain dynamical pathways toward stable and metastable morphologies, while monitoring discrete energy dissipation and constraint satisfaction.

\vspace{-10pt}
\subsection{Direct constrained energy minimization}\label{sec: trust_region}
We compute metastable equilibria by directly minimizing the full energy \eqref{eq:energy_full_nd}
in a finite-dimensional Fourier representation on the periodic material domain $\sigma\in[0,2\pi)$.
We approximate each field $F(\sigma)\in\{h,\kappa,c\}$ by a truncated Fourier series with $N_{\rm mode}$ modes,
\begin{equation}
F(\sigma)=a_0+\sum_{n=1}^{N_{\rm mode}}\big(a_n\cos n\sigma+b_n\sin n\sigma\big),
\end{equation}
and use the Fourier coefficients as optimization variables.

The resulting constrained nonlinear program with constraints \eqref{eq:closure}, \eqref{eq: turning}, and \eqref{eq:c_total_nd} is solved using SciPy's \verb|trust-constr| algorithm \cite{trust-constr}, which is designed for nonlinear equality/inequality constraints and is robust for our integral constraints.

To reduce missed minima across parameter space, we perform a neighbor-seeded continuation sweep on a parameter grid \cite{allgower2003introduction}: each grid point is re-optimized using converged solutions from nearby points as initial guesses, and we retain the lowest-energy solution. Iterating this sweep yields a self-consistent set of metastable branches and transition boundaries.

\vspace{-10pt}
\subsection{Time-dependent gradient flow simulations}\label{sec: pseudo}
We compute dynamical pathways toward metastable morphologies by numerically integrating the time-dependent gradient flow \eqref{eq:evolution_gkc}-\eqref{eq:vt_vn} using a pseudo-spectral method. 

Because the model is a high-order, strongly nonlinear system on a closed loop, the pseudo-spectral method is most suitable for the following reasons. 
First, on a periodic domain, derivatives reduce to simple multipliers $(ik)^n$ in Fourier space, which makes high-order operators easy to implement and highly accurate. This is particularly important when comparing morphologies and energy levels, where errors in high-order terms can bias the dynamics.
Second, the material coordinate $\sigma\in[0,2\pi)$ is naturally periodic, so the spectral representation enforces periodic boundary conditions automatically, avoiding explicit endpoint identification.

We integrate the coupled system \eqref{eq:evolution_gkc}--\eqref{eq:vt_vn} using a pseudo-spectral discretization in $\sigma$ and an IMEX second-order semi-implicit backward differentiation formula (SBDF2) time stepper \cite{ascher1995implicit, wang2008variable}, implemented in Dedalus \cite{burns2020dedalus}. At each step $t^n\to t^{n+1}$, we first advance the three fields $(h,\kappa,c)$ simultaneously with SBDF2. We then assemble the thermodynamic stress $\mathcal P$ \eqref{eq: P_stress} from the updated fields and apply a low-pass spectral filter to $\mathcal P$. Using the filtered stress, we compute the effective line tension $\beta(h-h_0)-\alpha\mathcal P$ and evaluate the velocities $(v_t,v_n)$ from \eqref{eq:vt_vn}. Finally, we apply a light low-pass filter to the assembled normal velocity $v_n$ before using $(v_t,v_n)$ in the next update of \eqref{eq:evolution_gkc}.  The evolution of velocities depend only on the state ($h, \kappa, c$) rather than the reconstructed embedding $\mathbf{x}$, so any closure defect in $\mathbf{x}$ does not feed back or accumulate among iterations.

This IMEX--spectral scheme is suited and efficient for the present high-order, strongly nonlinear dynamics. Implicit treatment of the stiff Cahn--Hilliard diffusion, together with targeted filtering of $\mathcal P$ and $v_n$, suppresses high-frequency contamination and improves long-time stability. The Fourier pseudo-spectral discretization resolves high-order derivatives accurately on the periodic loop, and the modular velocity assembly localizes stabilization to the most aliasing-prone terms. The method is also simple and readily extensible, serving as a baseline that can be upgraded to formal operator-splitting methods \cite{blanes2024splitting} if needed.

\section{RESULTS} \label{sec:results}
\begin{figure*}[t]
    \centering

    \begin{subfigure}[t]{0.48\linewidth}
        \centering
        \includegraphics[width=\linewidth]{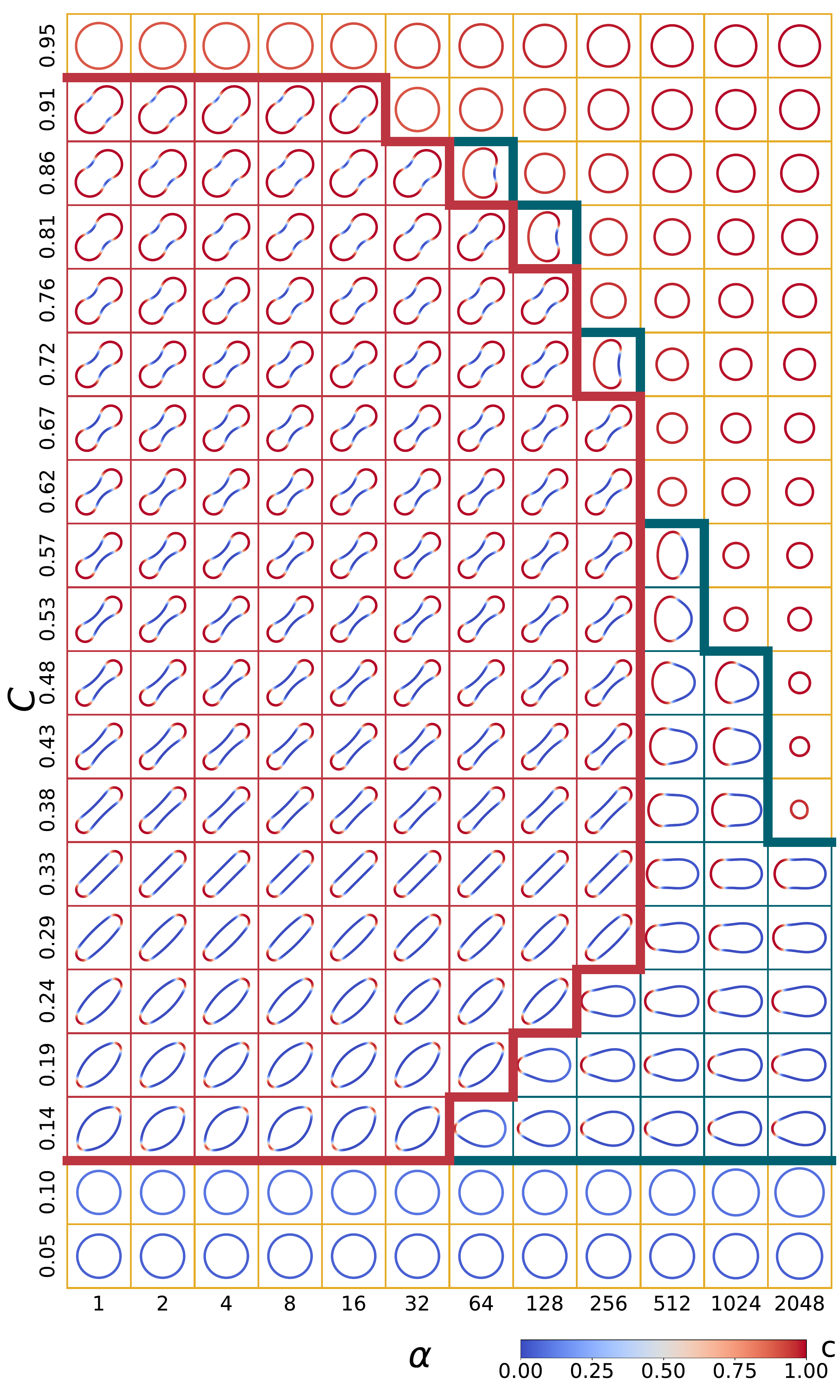}
        \caption{$\varepsilon=0.05$}
        \label{fig:phase_ep05}
    \end{subfigure}\hfill
    \begin{subfigure}[t]{0.48\linewidth}
        \centering
        \includegraphics[width=\linewidth]{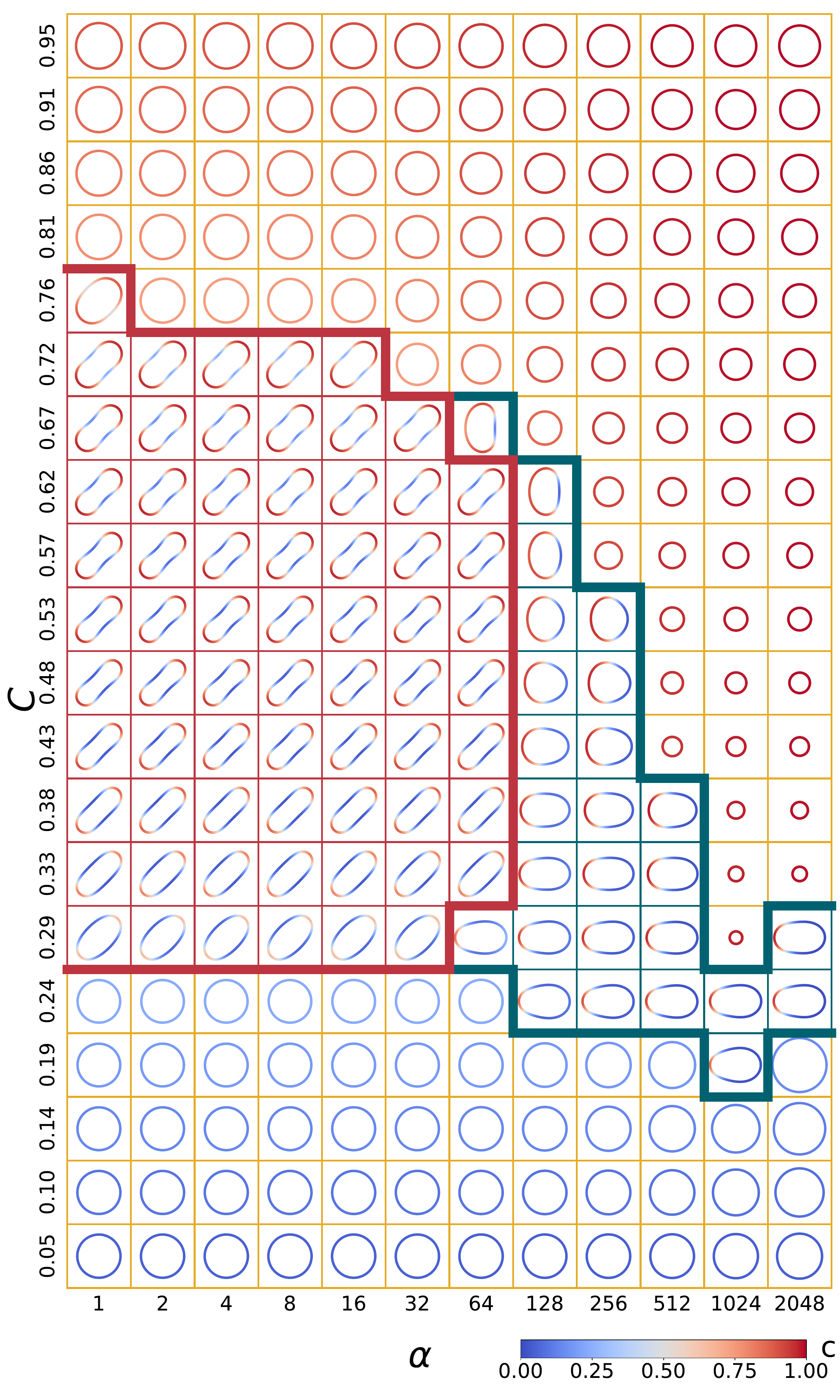}
        \caption{$\varepsilon=0.15$}
        \label{fig:phase_ep15}
    \end{subfigure}
    \caption{ Phase diagrams in the ($\alpha, C$) plane of minimizer morphologies, for two values of the interfacial parameter $\varepsilon=0.05, 0.15$. Within each morphology, the red-blue color scale depicts the local composition $c$. The solid red and green lines delineate phase boundaries of the $N=4$ (peanut) and $N=2$ (acorn) phases; outside these (yellow shading) the minimizer is a compositionally homogeneous circle ($N=0$). Increasing $\varepsilon$ raises the interfacial cost and shifts the optimum toward states with fewer interfaces. Further parameters: $\kappa_0 = 3$, $\beta=20$;  the $C$-axis ticks shown are rounded values of a uniform grid, exact values and} numerical details are in Appendix \ref{secapp: numerical}.
    \label{fig:phase_diagrams_alphaC}
\end{figure*}

Since our model does not include noise, it leads to morphologies that correspond to dynamically stable steady states of the deterministic gradient-flow dynamics, {\em i.e.}, constrained local minima of the free-energy functional. We classify the outcomes in terms of \textit{morphological phases}, distinguished by discrete geometric diagnostics, most importantly the interface number $N$. Phase boundaries in the diagrams indicate parameter values where competing stable morphologies exchange global optimality via an energy crossing, giving an abrupt (often discontinuous) change of shape and frequently a jump in $N$.

\subsection{Control parameters }
The nondimensional parameter set $\left\{\alpha, \beta, \varepsilon, \kappa_0, m, C\right\}$, was define in Section \ref{sec: nondim} above. We generally fix $m=1$ in this work; the morphology is then controlled primarily by  $\alpha, \varepsilon, \kappa_0$, and $C$, which together determine the competition between interfacial cost and curvature-composition mismatch. In particular, $\alpha$ and $\varepsilon$ penalize the creation of interfaces and thus favor smaller domain number $N$, whereas $\kappa_0$ and $C$ control how easily the high-density phase can support the required turning number and reduce the bending mismatch. The remaining parameter, $\beta$ controls stretching. As emphasized in the introduction we are primarily interested in the near-inextensible limit and so take $\beta=20$.
In much of the parameter space the metric then remains almost constant, 
$E_{\rm metric}$ is subdominant, and morphology is controlled primarily by the  competition between phase separation and bending.
At sufficiently large $\alpha$ (and $\varepsilon$), however, the field-induced stresses can cause
stretching to contribute nontrivially to morphology selection for the chosen $\beta$ value as we shall see below. 

\subsection{Minimizer phase diagrams} \label{sec: phase_diagram}
We have explored morphologies by free energy minimization on various two-dimensional cuts through the parameter space, as detailed below.

\subsubsection{Phase diagrams on the \texorpdfstring{$(\alpha, C)$}{(alpha, C)} plane.}
Figure~\ref{fig:phase_diagrams_alphaC} maps the globally minimizing morphologies (classified by the interface number, $N$) in the $(\alpha,C)$ plane for two different interfacial widths $\varepsilon$  corresponding to relatively sharp and relatively broad interfaces on the scale of the filament length. The observed structure is governed by a competition between (i) interfacial cost (favoring interface number $N=0$ or $N=2$)), and (ii) curvature--composition mismatch (often favoring larger $N$). The latter can be reduced by redistributing the turning budget across multiple enriched segments when it cannot be efficiently accommodated by a single dense domain. 

The boundaries separating the $N=0$, $N=2$, and $N=4$ regions correspond to loci where morphologies coexist as local minima and exchange global optimality via an energy crossing, producing an abrupt change in the identity of the minimizer. Consistent with the interfacial-cost picture, increasing $\varepsilon$ systematically shrinks the $N=4$ and $N=2$ region and expands the parameter range where $N=0$ is optimal. Note that at the boundaries of the latter range (shaded yellow) stretching plays an increasingly important role as the coupling constant $\alpha$ is increased. This is because any phase-separated state can in principle avoid paying interfacial costs by remaining uniform, whereupon shrinking or expanding the filament moves the global density  closer to one of the binodals ($c = \pm 1$), thereby reducing the chemical cost of the uniform state. For large $\alpha$, where minimization of $E_{\rm field}$ takes priority over other energy contributions, this becomes an advantageous strategy. At small $\alpha$ the length of the filament is effectively fixed; here the transition from the $N=4$ morphology to $N=0$ occurs when the global composition is such that there is insufficient minority phase to create well-separated interfaces. This criterion depends on $\varepsilon$; the latter dependence is responsible for the differences, between the two phase diagrams shown, in where these boundaries lie. It is important to note that stretching does not affect the boundaries between different morphologies with $N\ge 2$; as quantified in Table \ref{table:metastables}, stretching plays only a small role in the energetics of all such morphologies.

\begin{figure}[h]
    \centering
    \includegraphics[width=1\linewidth]{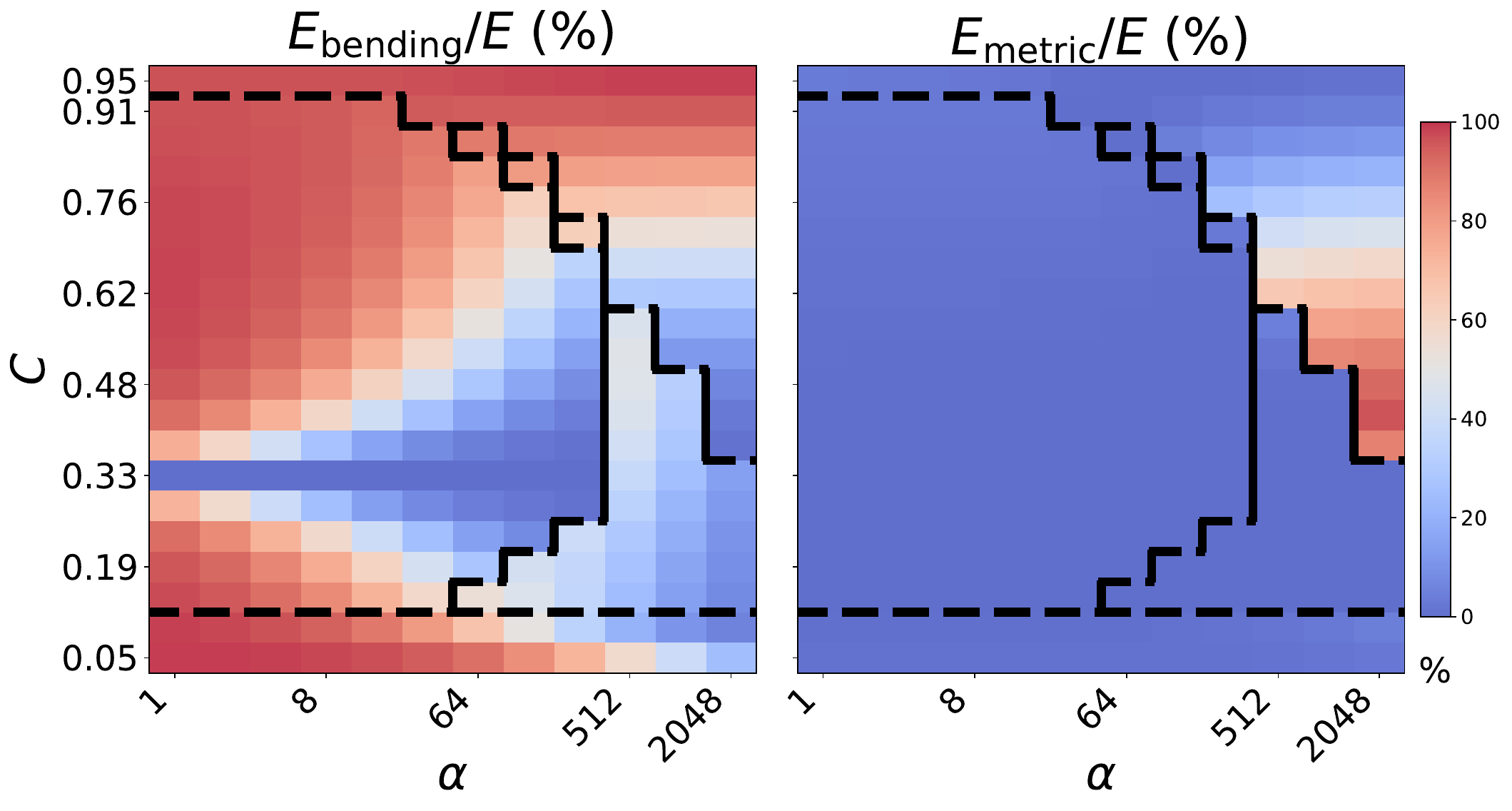}
    \caption{
    Heatmaps of the energy fractions $E_{\text {bending }} / E$ (left) and $E_{\text {metric }} / E$ (right), evaluated at the minimum-energy state across the ($\alpha, C$) phase plot \ref{fig:phase_ep05} with phase boundaries shown as dashes. The left figure shows a minimal bending energy proportion near $1/\kappa_0$ for the $N=4$ metastable shapes.
(Numerical details in Appendix \ref{secapp: numerical}.)}
    \label{fig:energy_decomp}
\end{figure}

The energetics behind the phase diagrams on the $(\alpha,C)$ plane is further explored by the energy decomposition in Fig.~\ref{fig:energy_decomp}. This shows the phase boundaries for the case $\varepsilon = 0.05$ (left panel in Fig.~\ref{fig:phase_diagrams_alphaC}) with overlaid heat maps for the bending and stretching contributions. The bending energy shows a sharp downward feature at $C = C^* = 1/\kappa= 0.33$ which is the global composition where the turning constraint is satisfiable without bending cost. Bending dominates the stretching energy elsewhere in the phase diagram, except for a region close to the boundary between $N=0$ and $N\ge 2$ states which occurs at high $\alpha$. Here  there is a tradeoff between shrinkage and bending as described already above.

For $C$ near $C^*$, the turning constraint $\int_\Gamma \kappa\,ds=2\pi m$ is most compatible with the local curvature
preference $\kappa\approx \kappa_0 c$, so that curvature--composition mismatch is subdominant (cf.~\eqref{eq: energy_sharp}).
In the near-inextensible regime this matching condition gives $\kappa_0\int_\Gamma c\,ds \approx 2\pi m$, hence
$C^*=m/\kappa_0$ (so $C^*=1/\kappa_0$ for $m=1$). 

When $C<C^*$, there is insufficient composition to provide the required turning angle by staying close to the preferred curvature, $\kappa \approx \kappa_0 c$. The filament must then generate substantial curvature in regions where $c \approx 0$ where the spontaneous curvature is zero. This necessarily increases the curvature-composition mismatch $\kappa- \kappa_0 c$. As a result, the elastic part of the energy budget becomes relatively larger. 

\begin{figure*}[t]
  \centering
    \begin{subfigure}[t]{0.75\linewidth}
        \centering
        \includegraphics[width=\linewidth]{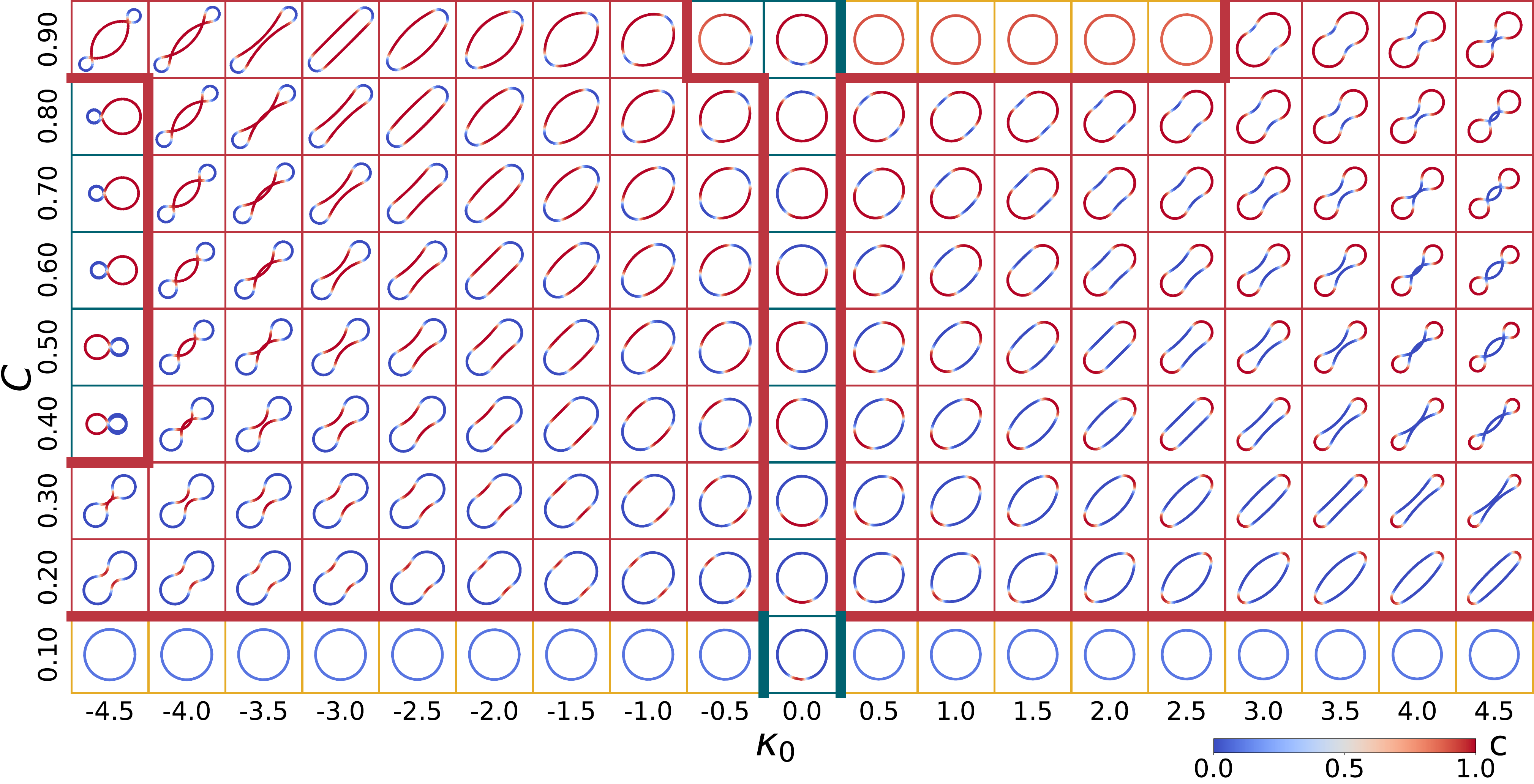}
        \caption{$\alpha=1$}
        \label{fig:phase_a1}
    \end{subfigure}
    \\
    \begin{subfigure}[t]{0.75\linewidth}
        \centering
        \includegraphics[width=\linewidth]{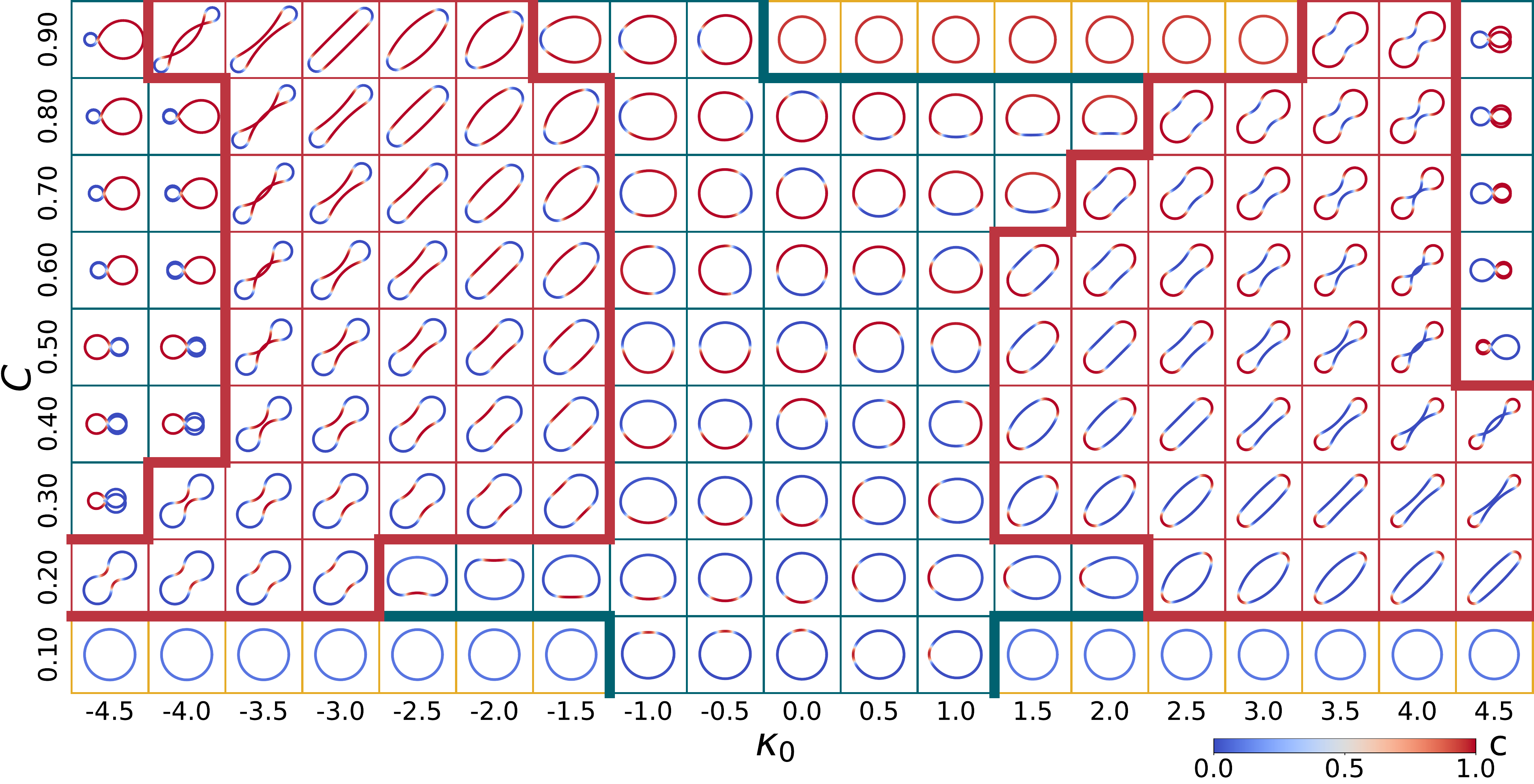}
        \caption{$\alpha=64$}
        \label{fig:phase_a2}
    \end{subfigure}
    \\
    \begin{subfigure}[t]{0.75\linewidth}
        \centering
        \includegraphics[width=\linewidth]{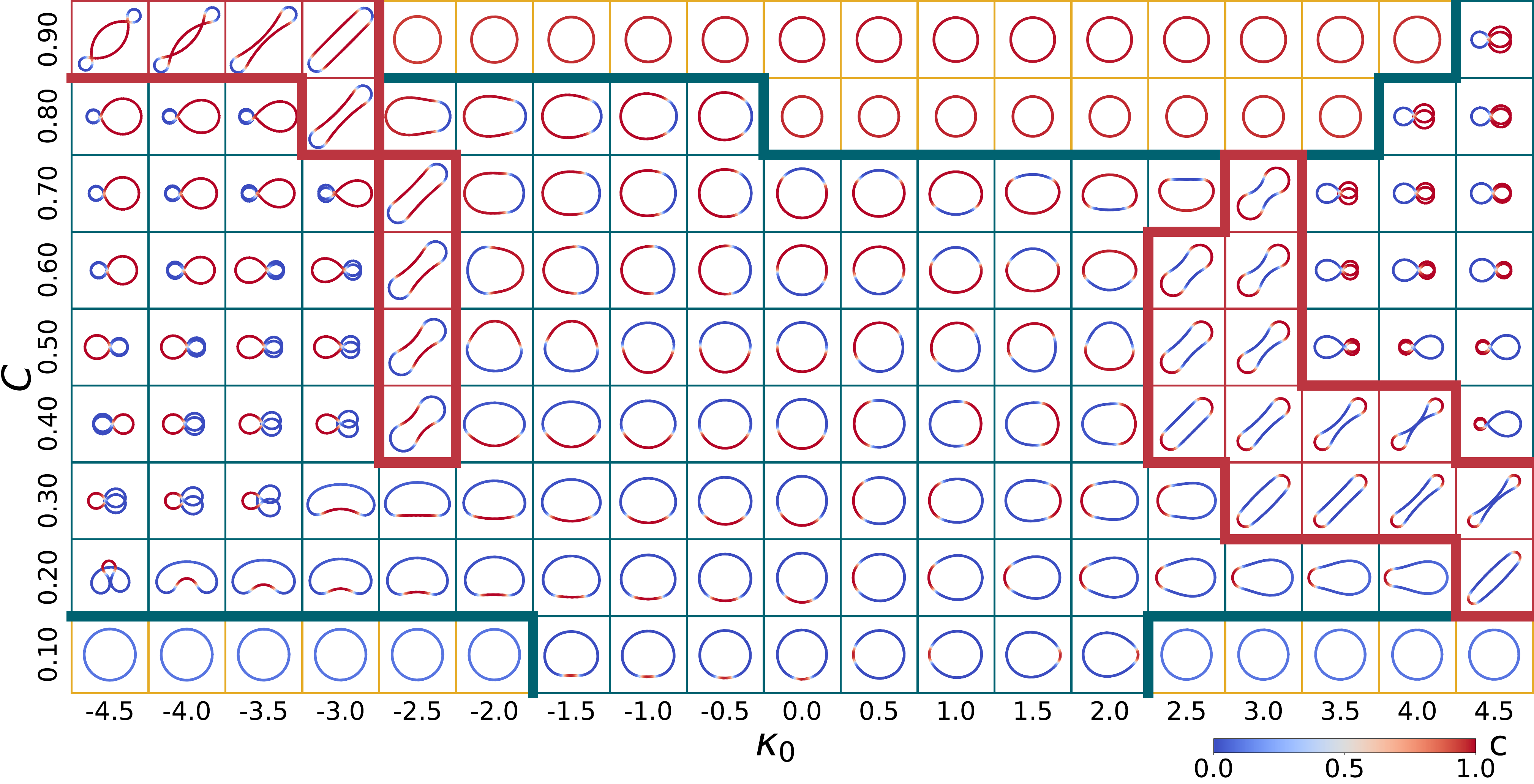}
        \caption{$\alpha=256$}
        \label{fig:phase_a3}
    \end{subfigure}  
    \caption{Global-minimizer morphology maps in the $(C,\kappa_0)$ plane for (a--d) $\alpha=1,64,256$ (other parameters fixed as in Fig.~\ref{fig: metastables}). Each cell displays the reconstructed minimum-energy configuration, colored by the concentration field $c$. Thick lines delineate the globally optimal morphologies classified by interface number: $N=0$ (shaded yellow), $N=2$ (green line), and $N=4$ (red line).}
    \label{fig:phase_diagrams_k0C}
\end{figure*}
Conversely, when $C>C^*$, the composition will provide more spontaneous curvature than is compatible with the fixed total turning. The system must then offset this by under-curving the dense phase and/or by having negative curvature in the dilute one.  

Notably, the bending-energy fraction exhibits a sharp jump across the $N=2\leftrightarrow N=4$ boundary (dashed line), providing an energetic signature of the morphology switch. This reflects an exchange of optimality between two distinct branches: an $N=2$ branch that saves interfacial cost but pays substantial curvature mismatch, and an $N=4$ branch that introduces additional interfaces to strongly reduce bending frustration.

\subsubsection{Phase diagrams on the \texorpdfstring{$(\kappa_0, C)$}{(kappa_0, C)} plane.}

Fig.~\ref{fig:phase_diagrams_k0C} maps the globally minimizing morphology in the $(C,\kappa_0)$ plane for several values of $\alpha$ (other parameters fixed as in Fig.~\ref{fig: metastables}). As $\alpha$ increases, the region where multi-domain states are globally optimal is progressively reduced, consistent with the higher energetic cost of additional interfaces. For $\kappa_0>0$, varying $\kappa_0$ primarily shifts the curvature-matching composition $C^*=1/\kappa_0$; this moves the location of the low-mismatch corridor where the turning constraint can be accommodated with little curvature--composition frustration. Increasing $|\kappa_0|$ enhances the energetic value of converting compositional contrast into curvature changes, and thereby promotes symmetry-breaking, phase-separated morphologies over uniform states.

The phase diagrams are generically not symmetric under $\kappa_0\to-\kappa_0$. This asymmetry is physically expected: the turning constraint fixes $\int_\Gamma \kappa\,ds=2\pi m>0$, so negative spontaneous curvature competes with the required positive turning. Moreover, because $c\in[0,1]$, the coupling term $\kappa_0 c$ provides bounded curvature control and does not allow a sign-symmetric redistribution of curvature between phases. In particular, the matching condition $C^*=m/\kappa_0$ lies outside the admissible range when $\kappa_0<0$, so curvature mismatch cannot be removed by tuning $C$ on the negative-$\kappa_0$ side. This factor reshapes the set of competing minima and can increase metastability.

\FloatBarrier
Because the present model does not include self-avoidance or contact forces, the reconstructed planar embedding may self-intersect in regimes of sufficiently strong curvature forcing (large $|\kappa_0|$). In
Fig.~\ref{fig:phase_diagrams_k0C} such solutions appear as overlapping `figure-8' shapes. For the $N=4$ branch in the sharp-interface limit, an explicit criterion for the onset of such curve-crossing can be derived from a symmetric four-arc construction.

This criterion successfully identifies the onset of overlap observed in the corresponding region of parameter space. For $N=2$ we also observe self-intersection numerically, but do not have a comparably simple closed-form condition.

\subsection{Coarsening dynamics and metastability}

\begin{figure*}
  \centering
  \begin{subfigure}[t]{0.45\textwidth}
    \centering
    \includegraphics[width=\textwidth]{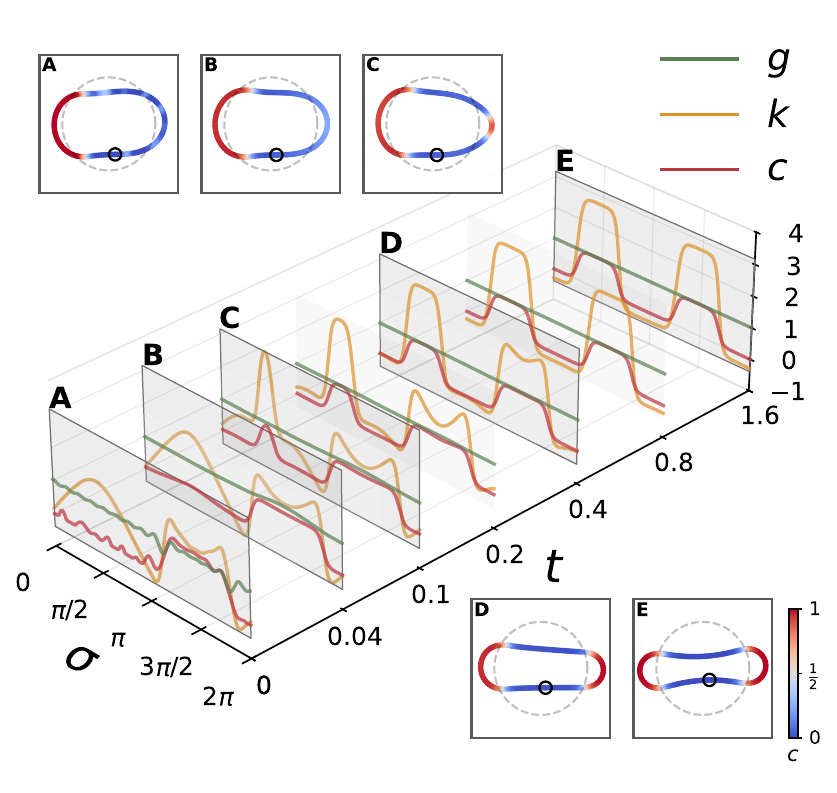}
       \label{fig:PDE_solution_left}
  \end{subfigure}
  \hfill 
  \begin{subfigure}[t]{0.45\textwidth}
    \centering
    \includegraphics[width=\textwidth]{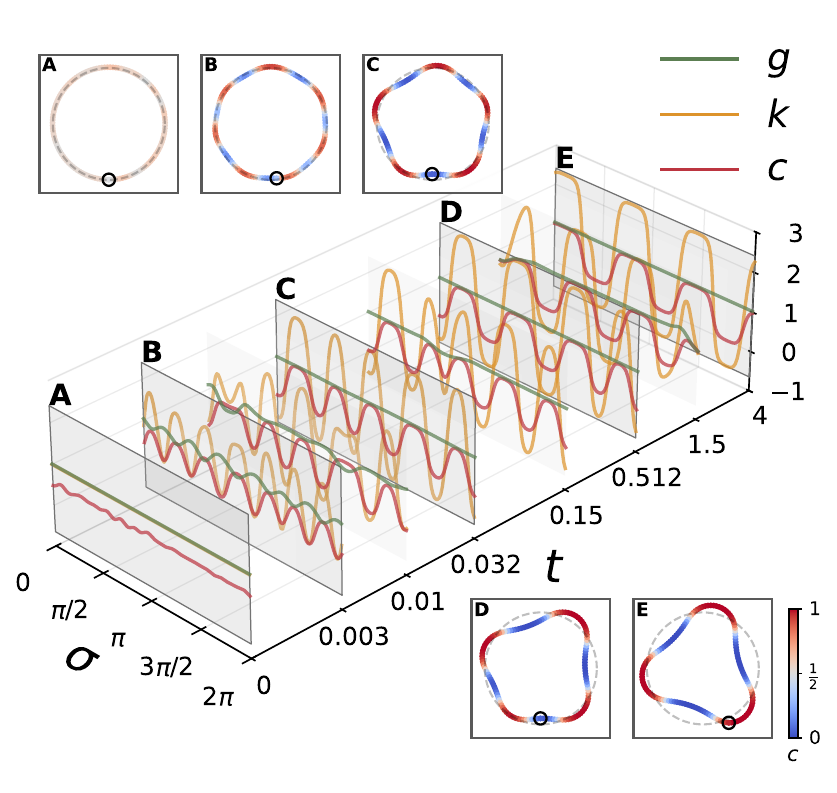}
    \label{fig:PDE_solution_right}
  \end{subfigure}
  \\[0.0em]
  \begin{subfigure}[b]{0.49\textwidth}
      \includegraphics[width=1\textwidth]{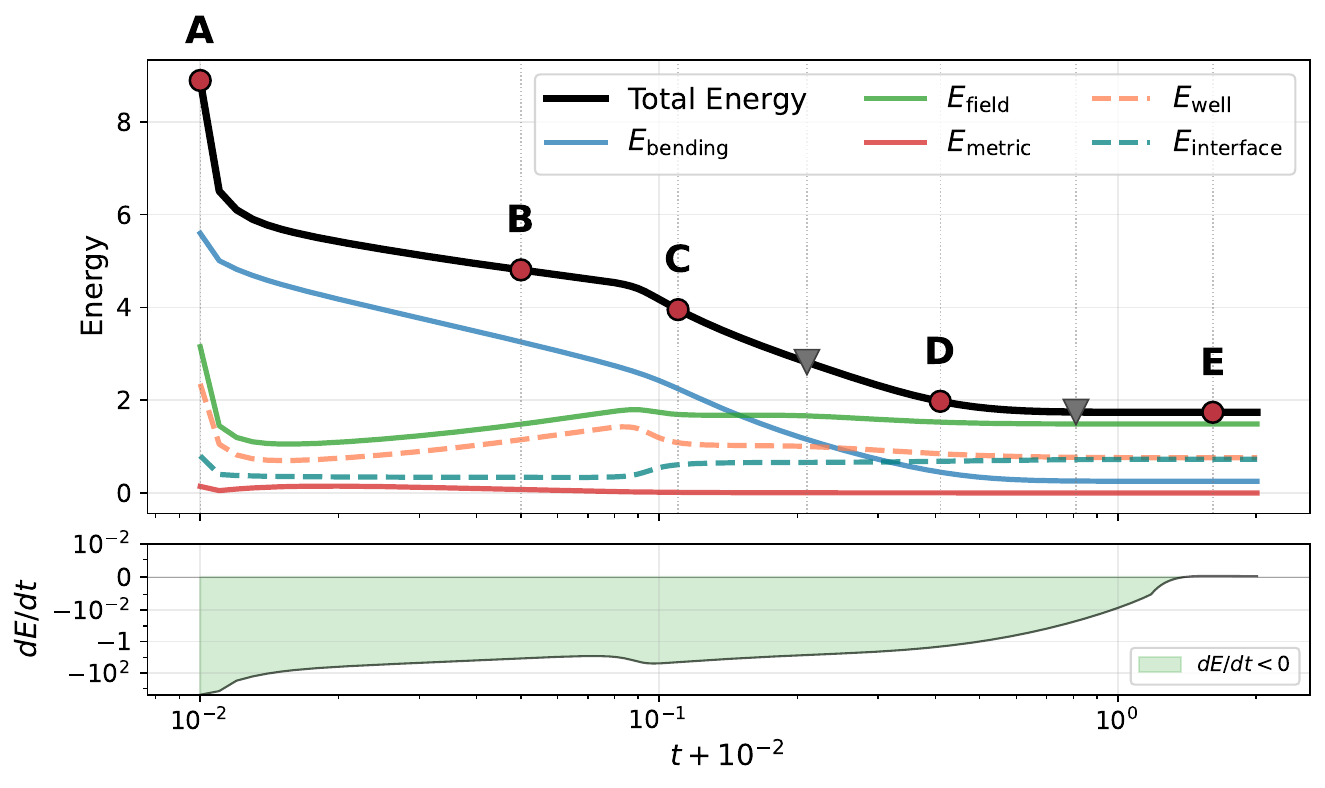}
      \label{fig:energy_component_evolution_left}
  \end{subfigure}
    \hfill 
  \begin{subfigure}[b]{0.49\textwidth}
      \centering
      \includegraphics[width=1\textwidth]{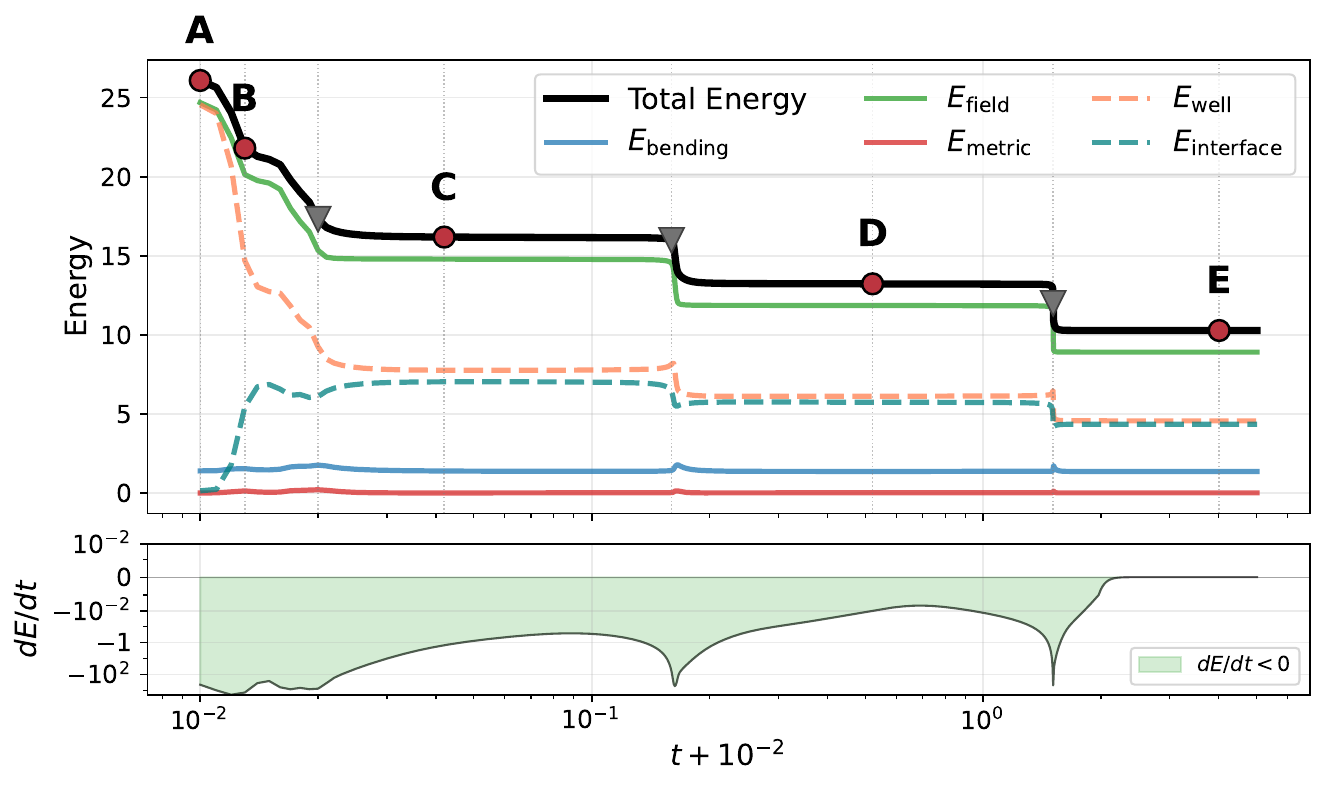}
      \label{fig:energy_component_evolution_right}
  \end{subfigure}
  \caption{
  Time-dependent gradient flow \eqref{eq:evolution_gkc} \eqref{eq:vt_vn}pathways illustrated by waterfall plots. Waterfall plots show the evolution of the concentration $c$ (red), curvature $\kappa$ (orange), and metric $g$ (green) along the material coordinate $\sigma$ at a sequence of times $t$ under the gradient flow simulations of \eqref{eq:evolution_gkc}-\eqref{eq:vt_vn}. Insets (A-E) display the corresponding reconstructed loop shapes, color-coded by $c$; the gray dashed circle indicates the relaxed reference shape without coupling, and a small black circle indicates the position $\sigma=0$.  Parameter data and initialization details (involving a small perturbation to the concentration field $c$) are given in Appendix \ref{secapp: numerical}. {\em Left:} Starting from an acorn shape ($N=2$), the evolution to a peanut-shape ($N=4$) morphology is shown. {\em Right:} With starting conditions close to circular state of uniform composition ($N=0$), the evolution to a polygonal ($N=6$) multi-domain morphology is shown. Domains merge at around $t=0.01, 0.15, 1.5$. Merge events take longer once there are fewer domains.}
    \label{fig:PDE_solution}
\end{figure*}
Time-dependent simulations of the gradient flow provide an independent route to the same equilibria as found by direct minimization, and reveal their basins of attraction. In our deterministic (noiseless) setting, metastable morphologies correspond to genuine local minima: they persist under the dynamics and are reached robustly from a range of initial conditions rather than representing incomplete relaxation. We have used the fully dynamics to check qualitative features of the phase diagrams presented earlier, but their computational cost prevents us doing so for more than a sparsely sampled set of phase points. 

In several parameter regimes we observe long-lived multi-domain states ($N\ge 4$) that persist over the longest simulated times, indicating strong metastability with finite energy barriers between competing minima. Such barriers do not arise for phase ordering on a rigid circular filament because there is in that case always a small attractive force between neighboring interfaces. This means that any stationary state with $N>2$ is locally unstable to perturbations that displace these interfaces. To understand why curvature-coupling breaks this condition so that coarsening can stall, we here analyze representative merging events via the time evolution of the total energy and its decomposition into bending, stretching, and composition-field/interfacial contributions. 

Figure \ref{fig:PDE_solution} shows two representative coarsening/merging episodes together with the evolution of the total energy and its main components. Merging reduces interfacial cost, producing a clear drop in the interfacial contribution, but it can transiently increase curvature mismatch and/or metric strain, leading to compensating rises or overshoots in bending/metric energies. Whether merging proceeds is therefore decided by a trade-off: the interfacial energy gain must outweigh the elastic penalties not just in the final state, but progressively at all points along the pathway. Whenever this does not hold, the system becomes trapped at a metastable fixed point.

The two lower panels in Fig.~\ref{fig:PDE_solution} show that the total free energy decreases monotonically throughout the time evolution, confirming that the numerical dynamics is accurately dissipative as required by the gradient flow structure of the governing equations. Moreover, the decomposition of the energy into its various contributions reveals a clear two-stage relaxation mechanism, comparable to that seen for coarsening in a fixed Euclidian domain.
At early times the system rapidly reduces the energy associated with the field sector (notably $E_{\rm field}$ which includes its bulk double-well part $E_{\rm well}$). This corresponding to fast early-stage phase separation toward nearly pure high- and low-density regions. This stage is marked by a steep initial drop in $E$. At later times, once the bulk phases are established, relaxation becomes controlled by coupled interface-geometry reorganization: the remaining energy decreases mainly through changes in the interfacial contribution $E_{\rm interface}$ together with compensating adjustments in the geometric energies $E_{\rm bending}$ (and, when extensibility matters, $E_{\rm metric}$). These features are caused by interfacial merger (or creation) events which can be separated by periods of slower evolution. 

Crucially, the late-time dynamics is no longer a purely local diffusive coarsening process as in Cahn-Hilliard on a fixed 1D domain; instead it requires nonlocal shape readjustments of the closed filament that redistribute curvature while maintaining the global geometric constraints.
Indeed, translational motion of an interface generally requires a loop-wide shape readjustment to maintain the closure and turning constraints while managing the curvature--composition mismatch. This nonlocal coupling effectively increases dynamical resistance to coarsening; it can strongly slow interfacial mergers or, when finite barriers arise, stop them entirely. In the former case, relaxation typically proceeds through long quasi-steady plateaus separated by discrete reconfiguration events.

The right-hand panel in Fig. \ref{fig:PDE_solution} provides direct evidence for such kinetically hindered coarsening. The total energy exhibits a sequence of step-like drops, each accompanied by a sharp negative spike in $d E / d t$  (interface mergers). Each merger is driven by the reduction of interfacial cost but it simultaneously incurs geometric penalties. The global geometric coupling can transiently raise $E_{\rm bending}$ and/or $E_{\rm metric}$ during a merger, reflecting a pathway-level trade-off between interfacial gain and elastic cost. As coarsening progresses and fewer interfaces remain, further mergers require increasingly global rearrangements, consistent with the growing waiting times between successive steps. The left-hand panel shows a different regime, involving the evolution from an unstable $N=2$ state to a stable $N=4$ morphology. Here the total energy decreases smoothly, although its components are non-monotonic, until all the energies saturate at plateau values, indicating approach to a stable fixed point of the dynamics.

\section{CONCLUSION}  \label{sec: outlook}
In this work we have studied a closed planar elastic filament carrying a conserved phase separating field, coupled to bending and stretching elasticity through a composition-dependent spontaneous curvature. The central physical finding is that global closure and topology fundamentally influence phase separation by enforcing nonlocal geometric constraints: on a closed loop, even a small displacement of an interface between coexisting phases generally requires a global readjustment of the curvature distribution to maintain closure and the fixed turning number. This `closure-induced curvature frustration' qualitatively changes the free energy landscape compared with phase separation on either a fixed periodic domain or an open filament, where coarsening generically proceeds toward a state with either two or one interfaces. In contrast, the closed, deformable setting admits distinct stable and genuinely metastable morphologies comprising persistent multi-domain states, including the minimal closure-compatible $N=4$ family. These exist alongside uniform ($N=0$) and two interface ($N=2$) branches and exchange global optimality via sharp energy crossings, at first-order morphological transitions. We further identify a natural `matching' coverage $C^* \sim m / \kappa_0$ at which the required total turning angle matches the local preference $\kappa \approx \kappa_0 c$, producing a pronounced minimum in curvature-composition mismatch (and thus bending cost). Overall, our results show that closure acts as an effective long-range interaction between phase boundaries, generating robust morphology selection and geometrically obstructed coarsening in an otherwise unfrustrated gradient flow system.

We believe our approach to equilibrium phase separation on a closed elastic filament can be extended in a systematic way to address active field theories. In such an extension, the scalar field $c$ would retain its role as a conserved order parameter but would be interpreted more broadly ({\em e.g.}, as a local occupancy, connectivity, or remodeling state of an interfacial scaffold). Activity can then be introduced by allowing (i) nonequilibrium currents as in Active Model B type coarse-grained descriptions \cite{wittkowski2014scalar, nardini2017entropy} and/or (ii) source-  sink terms or reaction processes that break detailed balance and prevent relaxation to a free-energy minimum \cite{GoychukFrey2019, gao2026self}. Coupling these active field dynamics to interfacial mechanics would provide a principled route to capturing (evolving or time-independent) nonequilibrium behaviors on closed geometries, and we hope to address these aspects in future studies. We anticipate that the our development of robust numerical methods for the full dynamics will prove crucial in this context where the alternative path of direct free energy minimization does not exist.

\section{ACKNOWLEDGMENTS}
HW is grateful to Darryl Holm and Boris Khesin for insightful discussions that helped shape this study. We also thank Mingnan Ding, Mingrui Ding, Antonio Filieri, Nir Gov, Kaibo Hu, Ruiao Hu, Xiaotong Ji, Ran Jing, Buyang Li, Rui Ma, Balazs Nemeth, Yitong Sun, Mohamed Warda and Antonia Winter for helpful discussions.

\section*{DATA AVAILABILITY STATEMENT}
The data that support the findings of this study are available from the corresponding author upon reasonable request. 

\bibliography{reference}

\onecolumngrid  

\newpage

\appendix
\newpage

\section{Variational Analysis}\label{secapp: variation}
\subsection{Variation of Geometry}
For a planar curve parameterized by arc length $s$, the Frenet frame $(\mathbf{T}, \mathbf{N})$ satisfies
\begin{equation}
\partial_s \mathbf{T}=\kappa \mathbf{N}, \quad \partial_s \mathbf{N}=-\kappa \mathbf{T},
\end{equation}
and a small deformation is written as
\begin{equation}
\delta \mathbf{x}=\xi \mathbf{T}+\eta \mathbf{N} .
\end{equation}
We choose the normal orientation such that $\eta>0$ corresponds to an inward normal perturbation, consistent with $\partial_s \mathbf{T}=\kappa \mathbf{N}$ and $\partial_s \mathbf{x} = \mathbf{T}$.\\

{\bf Variation of the metric:}

Considering $g=\left|\partial_\sigma \mathbf{x}\right|^2$ and  $\partial_s=g^{-1 / 2} \partial_\sigma$,  we have
\begin{equation}
\delta g
=2 \partial_\sigma \mathbf{x} \cdot \partial_\sigma(\delta \mathbf{x})
=2(\sqrt{g} \mathbf{T}) \cdot \sqrt{g}\left[\left(\partial_s \xi-\kappa \eta\right) \mathbf{T}+\left(\partial_s \eta+\kappa \xi\right) \mathbf{N}\right] 
=2 g\left(\partial_s \xi-\kappa \eta\right).
\end{equation}
Therefore, 
\begin{equation}
 \delta h = \delta \sqrt{g}=h\left(\partial_s \xi-\kappa \eta\right),\quad \delta\left(\partial_s\right)=-\left(\partial_s \xi-\kappa \eta\right) \partial_s, \quad
 \delta d s=\left(\partial_s \xi-\kappa \eta\right) d s\,.
\end{equation}

{\bf Variation of curvature:}

To obtain the variation of curvature $\delta \kappa$, we use $\partial_s \mathbf{T}=\kappa \mathbf{N}$:
\begin{equation} \label{eq: delta_k_prepare}
\delta\left(\partial_s \mathbf{T}\right)
=\delta\left(\partial_s\right) \mathbf{T}+\partial_s(\delta \mathbf{T})
=(\delta \kappa) \mathbf{N}+\kappa \delta \mathbf{N}.
\end{equation}
First, we compute the variation of vectors $ \delta \mathbf{T}$ and obtain the $ \delta \mathbf{N}$ by orthonormality:
\begin{equation}
\begin{aligned}
    \delta \mathbf{T}&=\delta\left(\partial_s\right) \mathbf{x}+\partial_s(\delta \mathbf{x}) = \left(\partial_s \eta+\kappa \xi\right) \mathbf{N}\\
    \delta \mathbf{N}&=-\left(\partial_s \eta+\kappa \xi\right) \mathbf{T} \,.
\end{aligned}    
\end{equation}
Next, we compute the other terms in \eqref{eq: delta_k_prepare} by using $\mathbf{T}$, 
\begin{equation}
\begin{aligned}
\delta\left(\partial_s\right) \mathbf{T}&=-\left(\partial_s \xi-\kappa \eta\right) \partial_s \mathbf{T}=-\left(\partial_s \xi-\kappa \eta\right) \kappa \mathbf{N}, \\
\partial_s(\delta \mathbf{T})&=\partial_s\left[\left(\partial_s \eta+\kappa \xi\right) \mathbf{N}\right]=\left(\partial_s^2 \eta+\partial_s(\kappa \xi)\right) \mathbf{N}-\left(\partial_s \eta+\kappa \xi\right) \kappa \mathbf{T} \,.
\end{aligned}
\end{equation}
Comparing the $\mathbf{N}$-components then yields
\begin{equation}
\delta \kappa=\partial_s^2 \eta+\kappa^2 \eta+\kappa_s \xi\,.
\end{equation}
These variation results are confirmed to be consistent with the kinematic form 
\begin{equation}
\begin{aligned}
& \partial_t g=2 g\left(\partial_s v_t-\kappa v_n\right) \\
& \partial_t k=\left(\partial_s^2+\kappa^2\right) v_n+\left(\partial_s \kappa\right) v_t
\end{aligned}
\end{equation}
by substituting $(\xi, \eta) \leftrightarrow\left(v_t, v_n\right)$. It is convenient to use $h=\sqrt{g}$ to avoid writing square roots, so that 
\begin{equation}
    \partial_t h=h\left(\partial_s v_t-\kappa v_n\right).
\end{equation}

\subsection{Compositional Variation and Cahn-Hilliard Flow}

We consider the nondimensional energy (see Appendix \ref{secapp: nondimension})
\begin{equation}
\mathcal{E}[\mathbf{x},c]
=\frac{1}{2}\int_\Gamma (\kappa-\kappa_0 c)^2\,ds
+\frac{\beta}{2}\int_{S^1}(h-h_0)^2\,d\sigma
+\alpha\int_\Gamma\Big[\,W(c)+\frac{\varepsilon^2}{2}\,c_s^2\,\Big]\,ds,
\label{eq:energy_full_nd_appendix}
\end{equation}
where $ds=h\,d\sigma$ and $F:=\kappa-\kappa_0 c$. We decompose the first variation as
\begin{equation}
    \delta\mathcal{E}
=\int_\Gamma\big(F_t\,\xi+F_n\,\eta\big)\,ds
+\int_\Gamma G^{(c)}\,\delta c\,ds,
\end{equation}
so that the gradient flow for the curve reads $v_t=-F_t$, $v_n=-F_n$ in the non-dimensionalized overdamping limit. Combining all tangential and normal parts found below in \eqref{eq:stretch_gradients}, \eqref{eq: bend_gradient}, and \eqref{eq: field_gradient}, we have the total dynamics in the main text \eqref{eq:evolution_gkc}.\\

{\bf Geometric variations of composition:}

Using $\delta h=h(\partial_s\xi-\kappa\eta)$, $\delta ds=(\partial_s\xi-\kappa\eta)\,ds$, and
$\delta\kappa=\partial_s^2\eta+\kappa^2\eta+\kappa_s\xi$, we can get the variation of the field $c$. Assume particles are perfectly bound to the filament and do not pass through material markers. Equivalently, the particle number carried by each material interval $\left[\sigma_1, \sigma_2\right]$ is invariant under a virtual deformation. Hence
\begin{equation}
    \delta\left(\int_{\sigma_1}^{\sigma_2} \rho_L(\sigma) d \sigma\right)=0 \text { for all } \sigma_1, \sigma_2\,.
\end{equation}
Since the interval is arbitrary, this implies the pointwise constraint
\begin{equation}
    \delta \rho_L(\sigma)=0\,.
\end{equation}
Therefore, varying $c=\rho_L /\left(\rho_0 h\right)$ at fixed $\sigma$ gives
\begin{equation}
    \delta c=\frac{1}{\rho_0} \delta\left(\frac{\rho_L}{h}\right)=\frac{1}{\rho_0}\left(\frac{\delta \rho_L}{h}-\frac{\rho_L}{h^2} \delta h\right)=-\frac{\rho_L}{\rho_0 h} \frac{\delta h}{h}=-c \frac{\delta h}{h}= -c(\partial_s\xi-\kappa\eta).
\end{equation}

{\bf Stretching energy:}

Let $d\sigma$ be a fixed reference parameter and $ds=h\,d\sigma$ the current arclength element, with a given reference density $h_0(\sigma)$.
Consider
\begin{equation}
    \mathcal{E}_g=\frac{\beta}{2}\int_{S^1}(h-h_0)^2\,d\sigma,
\end{equation}
using $\delta h = h(\partial_s\xi-\kappa\eta)$ and integrating by parts in the $\xi$-term (on a closed curve) gives:
\begin{equation}
    \delta \mathcal{E}_g
= \beta\int_{S^1}\,(h-h_0)\,\delta h\,d\sigma
= \beta\int_\Gamma \,(h-h_0)\,(\partial_s\xi-\kappa\eta)\,ds
= \beta\int_\Gamma \Big[-\partial_s\!(h-h_0)\,\xi
-\,\kappa\,(h-h_0)\,\eta\Big]\,ds,
\end{equation}
hence the stretching shape gradients are
\begin{equation}
F_{t, \rm str}=-\partial_s\!\big[\beta(h-h_0)\big],
\qquad
F_{n, \rm str}=-\,\beta\,\kappa\,(h-h_0).
\label{eq:stretch_gradients}
\end{equation}

{\bf Bending energy:}

Let $G:=\kappa-\kappa_0c$ with preferred spontaneous curvature $\kappa_0$. Consider
\begin{equation}
    \mathcal{E}_{\rm bend}=\frac{1}{2}\int_\Gamma (\kappa-\kappa_0c)^2\,ds\,, \quad \delta \mathcal{E}_{\rm bend} =\int_\Gamma G\,(\delta\kappa-\kappa_0\delta c)\,ds+\frac{1}{2}\int_\Gamma G^2\,\delta ds\,.
\end{equation}
Using $\delta h=h(\partial_s\xi-\kappa\eta)$, $\delta\kappa=\partial_s^2\eta+\kappa^2\eta+\kappa_s\xi$, and $\delta c=-c\left(\xi_s-\kappa \eta\right)$, we get the tangential part $\xi$-terms and normal part $\eta$-terms,
\begin{equation}
\begin{aligned}
    \delta\mathcal{E}_{\rm bend}\big|_{\xi} &= \int_\Gamma \big(G\kappa_s+G\kappa_0 c \partial_s+\tfrac{1}{2}G^2\partial_s\big)\,\xi\,ds
=\int_\Gamma [G(\kappa_s-\kappa_0c_s-G_s)-\kappa_0 c G_s]\,\xi\,ds.\\
\delta\mathcal{E}_{\rm bend}\big|_{\eta} &= \int_\Gamma \Big[G_{ss}+\kappa^2G-G\kappa \kappa_0 c-\tfrac{1}{2}\kappa G^2\Big]\eta\,ds
\end{aligned}
\end{equation}
The tangential and normal bending shape gradients follow as:
\begin{equation}\label{eq: bend_gradient}
    \begin{aligned}
        F_{t, \rm bend}&=-\kappa_0 c G_s =-\kappa_0 c \partial_s(\kappa-\kappa_0c)\,,\\
        F_{n, \rm bend}&=G_{s s}+\frac{1}{2} \kappa G^2
=\partial_s^2\left(\kappa-\kappa_0 c\right)+\frac{1}{2} \kappa\left(\kappa-\kappa_0 c\right)^2\,.
    \end{aligned}
\end{equation}

{\bf Composition field energy:}

Considering now the field energy part, 
\begin{equation}
    \mathcal{E}_{\rm field}=\alpha \int_{\Gamma}\left[W(c)+\frac{\varepsilon^2}{2} c_s^2\right] d s \,,
\end{equation}
the variation reads
\begin{equation}
    \begin{aligned}
        \delta {{\cal E}_{{\rm{field}}}} 
 &=\alpha \int {\left[ {{W^\prime }(c) - {\varepsilon ^2}{c_{ss}}} \right]\delta cds}  + \int {\left[ {W(c) - \frac{{{\varepsilon ^2}}}{2}c_s^2} \right]} \delta ds\\
 &= \alpha\int\left(-c \mu_{\text {field }}+f-\varepsilon^2 c_s^2\right)\left(\partial_s \xi-\kappa \eta\right) d s\\
 &= -\alpha \int_{\Gamma}\mathcal{P}\left(\partial_s \xi-\kappa \eta\right) d s\,,
    \end{aligned}
\end{equation}
where $f:=W+\frac{\varepsilon^2}{2} c_s^2$ and $\mu_{\rm field}:=\frac{1}{\alpha} \frac{\delta \mathcal{E}_{\rm field}}{ \delta c}=W^{\prime}(c)-\varepsilon^2 c_{s s}$, and the \textbf{chemical stress} as the line-tension density:
\begin{equation}
    \mathcal{P}:=c \mu_{\rm field} - f +\varepsilon^2 c_s^2\; .
\end{equation}
Then the tangential and normal bending shape gradient are 
\begin{equation}\label{eq: field_gradient}
    F_{t, {\rm field}}=\alpha \partial_s \mathcal{P}, \quad F_{n, {\rm field}}=\alpha\kappa\mathcal{P}
\end{equation}

{\bf Compositional flux:}

In the absence of sources and sinks, particle number conservation on the fixed material domain
implies the Lagrangian continuity equation with the definition in \eqref{eq: rho_L} and \eqref{eq:def_c}:
\begin{equation}
\partial_t \rho_L + \partial_\sigma J_L = 0,
\label{eq:lagrangian_continuity_app}
\end{equation}
where $J_L(\sigma,t)$ is the particle number flux across a material point (units: \# per time).
We now introduce the normalized flux
\begin{equation}
    J:=\frac{J_L}{\rho_0}\,.
\end{equation}
Using $c=\rho_L/(\rho_0 h)$ and \eqref{eq:lagrangian_continuity_app},
\begin{equation}
    \partial_t c=\frac{1}{\rho_0}\partial_t\!\left(\frac{\rho_L}{h}\right)
=\frac{1}{\rho_0}\left(\frac{\partial_t\rho_{L}}{h}-\frac{\rho_L}{h^2}h_t\right)
=-\frac{1}{h}\partial_\sigma J - c\,\frac{h_t}{h}.
\end{equation}
Since $h_t/h=\partial_s v_t-\kappa v_n$ and $\partial_s=h^{-1}\partial_\sigma$, this becomes
\begin{equation}
\partial_t c = -\,\partial_s J \;-\; c\big(\partial_s v_t-\kappa v_n\big).
\label{eq:c_kinematic_flux_app}
\end{equation}
We close \eqref{eq:c_kinematic_flux_app} by prescribing a Cahn--Hilliard flux along the physical
arclength coordinate:
\begin{equation}
J=-M\,\partial_s\mu,
\qquad
\mu=\frac{\delta\mathcal{E}}{\delta c}
=\alpha\big(W'(c)-\varepsilon^2 c_{ss}\big)-\kappa_0\big(\kappa-\kappa_0 c\big),
\label{eq:CH_flux_app}
\end{equation}
so that
\begin{equation}
\partial_t c = -c\big(\partial_s v_t-\kappa v_n\big) + M\,\partial_s^2\mu.
\label{eq:c_evolution_app}
\end{equation}

With the time scale $T_0=\zeta L_0^4/\beta_\kappa$ used in Appendix~\ref{secapp: nondimension}, the
dimensionless mobility is $\tilde M:=M\zeta$, and the dimensionless flux law reads
\begin{equation}
    \tilde J=-\tilde M\,\partial_{\tilde s}\tilde\mu.
\end{equation}
In the main text we drop tildes; in particular, $M$ denotes the dimensionless mobility and $J$
the dimensionless normalized flux.

\section{Non-dimensionalization of Energy and Equations of Motion}\label{secapp: nondimension}

We start from the dimensional total energy of a closed elastic loop $\Gamma(t)\subset\mathbb{R}^2$ coupled to a scalar composition field $c$:
\begin{equation}
E[\mathbf{x},c]
=
\frac{\beta_{\kappa}}{2}\int_\Gamma (\kappa - \kappa_0 c)^2\,ds
+\frac{\beta}{2}\int_{S^1} (h-h_0)^2\,d\sigma
+\alpha \int_\Gamma \left[\, W(c) + \frac{\varepsilon^2}{2}\,|\partial_s c|^2 \right] ds .
\label{eq:dim_energy}
\end{equation}
Here $\mathbf{x}(\sigma)$, $\sigma\in S^1$, parameterizes the loop, with local stretch
$h(\sigma)=|\partial_\sigma \mathbf{x}(\sigma)|$ and physical arclength element $ds=h\,d\sigma$.
The curvature of the loop is denoted by $\kappa$, and $\partial_s$ is differentiation with respect to arclength. The field $c=c(s)$ is a dimensionless occupation (or composition) variable, normalized by a reference density, so that $0\le c\le 1$. 
\[
[\mathbf{x}] = \mathrm{length},\qquad
[\sigma]=1,\qquad
[s]=\mathrm{length},\qquad
[h]=\mathrm{length},\qquad
[\kappa]=\mathrm{length}^{-1},\qquad
[c]=1.
\]
The coefficients $\beta_{\kappa}$, $\beta$, and $\alpha$ represent the bending stiffness, stretching
stiffness, and mixing-energy scale, respectively. Their physical dimensions are
\[
[\beta_{\kappa}]=\mathrm{energy}\cdot \mathrm{length},\qquad
[\beta]=\mathrm{energy}/\mathrm{length}^2,\qquad
[\alpha]=\mathrm{energy}/\mathrm{length},\qquad
[\varepsilon]=\mathrm{length},\qquad
[\kappa_0]=\mathrm{length}^{-1}.
\]

\textbf{Choice of characteristic scales:}

To simplify the formulation and identify the essential control parameters, we introduce the
characteristic scales
\begin{enumerate}
    \item \emph{Length:} the reference total length
$L_0 := \frac{1}{2\pi}\int_{S^1} h_0(\sigma)\,d\sigma, $
    \item \emph{Energy:} the bending energy scale
    $E_0 := {\beta_{\kappa}}/{L_0},$
    \item \emph{Time:} the overdamped relaxation time 
    $T_0 := {\zeta L_0^4}/{\beta_{\kappa}}.$
\end{enumerate} 
We define dimensionless variables (denoted by tildes) by
\begin{equation}
\mathbf{x}=L_0\tilde{\mathbf{x}},\qquad
s=L_0\tilde{s},\qquad
h=L_0\tilde{h},\qquad
h_0=L_0\tilde{h}_0,\qquad
\kappa=\frac{\tilde{\kappa}}{L_0},\qquad
\kappa_0=\frac{\tilde{\kappa}_0}{L_0},\qquad
\varepsilon=L_0\tilde{\varepsilon}.
\label{eq:nd_scaling}
\end{equation}
Consequently,
\[
ds =hd\sigma=L_0\tilde{h}d\sigma= L_0\,d\tilde{s},\qquad
\partial_s = \frac{1}{L_0 \tilde{h}}\partial_{\sigma}=\frac{1}{L_0}\partial_{\tilde{s}}.
\]
Let $\tilde{\Gamma}:=\Gamma/L_0$ denote the rescaled dimensionless loop. Substituting
\eqref{eq:nd_scaling} into~\eqref{eq:dim_energy} and dividing by $E_0=\beta_{\kappa}/L_0$ yields the
dimensionless energy $\tilde{\mathcal{E}}:=E/E_0$:
\begin{equation}
\tilde{\mathcal{E}}[\tilde{\mathbf{x}},c]
=
\frac{1}{2}\int_{\tilde{\Gamma}} (\tilde{\kappa}-\tilde{\kappa}_0\,c)^2\,d\tilde{s}
+\frac{1}{2}\frac{\beta L_0^3}{\beta_{\kappa}}\int_{S^1} (\tilde{h}-\tilde{h}_0)^2\,d\sigma
+\frac{\alpha L_0^2}{\beta_{\kappa}}\int_{\tilde{\Gamma}}
\left[\, W(c) + \frac{\tilde{\varepsilon}^2}{2}\,|\tilde{h}^{-1}\partial_{\sigma} c|^2 \right] d\tilde{s}.
\label{eq:nd_energy_full}
\end{equation}
or 
\begin{equation}
\tilde{\mathcal{E}}[\tilde{\mathbf{x}},c]
=\frac{1}{2}\int_{S^1} (\tilde{\kappa}-\tilde{\kappa}_0\,c)^2\tilde{h}\,d\sigma
+\frac{\tilde{\beta}}{2}\int_{S^1} (\tilde{h}-\tilde{h}_0)^2\,d\sigma
+\tilde{\alpha}\int_{S^1}
\left[\, W(c) + \frac{\tilde{\varepsilon}^2}{2}\,|\tilde{h}^{-1}\partial_{\sigma} c|^2 \right] \tilde{h}d\sigma.
\label{eq:nd_energy}
\end{equation}
with the dimensionless parameters defined by
\begin{equation}
\tilde{\alpha} := \frac{\alpha L_0^2}{\beta_{\kappa}},\qquad
\tilde{\beta} := \frac{\beta L_0^3}{\beta_{\kappa}},\qquad
\tilde{\varepsilon} := \frac{\varepsilon}{L_0},\qquad
\tilde{\kappa}_0 := \kappa_0 L_0.
\label{eq:nd_params}
\end{equation}

\textbf{Reference stretch:}

The prescribed reference metric $h_0(\sigma)$ encodes the natural local stretch of the material.
In general, $h_0$ may vary along the loop and should be retained as a given function in the
dimensionless stretching energy. In the special case of a uniform reference configuration,
one may choose $L_0=\tfrac{1}{2\pi}\int_{S^1}h_0\,d\sigma$ and reparametrize $\sigma\in[0,2\pi)$ so that the
reference stretch becomes constant. With the present scaling \eqref{eq:nd_scaling}, this yields
$
\tilde{h}_0(\sigma)\equiv 1
$, which we adopt unless stated otherwise.\\

\textbf{Global constraints:}

The closed-loop geometry and the conserved composition impose two global constraints.
The turning number constraint
$
\int_{\Gamma}\kappa\,ds = 2\pi m
$
is invariant under rescaling, hence $\int_{\tilde{\Gamma}}\tilde{\kappa}\,d\tilde{s} = 2\pi m$.
Conservation of total composition gives
$C_0 := \int_{\Gamma(t)} c\,ds$ (independent of $t$)
and we define the dimensionless conserved coverage
\begin{equation}
C=\frac{1}{2 \pi} \int_{S^1} c(\sigma, t) \tilde{h}(\sigma, t) d \sigma\,.
\label{eq:mass_nd}
\end{equation}
After nondimensionalization, the governing system is parameterized by the four independent
dimensionless groups and the two conserved and topological quantities,
$\{\tilde{\alpha},\,\tilde{\beta},\,\tilde{\varepsilon},\,\tilde{\kappa}_0;\; m,\,{C}\}$.
Here $\tilde{\alpha}$ controls the mixing strength, $\tilde{\beta}$ the stretching stiffness,
$\tilde{\varepsilon}$ the interfacial width, and $\tilde{\kappa}_0$ the curvature--composition coupling,
while $m$ fixing the turning number and ${C}$ fixing the conserved total composition.\\

{\bf Time scale and dynamical rescalings:}

In addition to the length scale $L_0$ and energy scale $E_0=\beta_\kappa/L_0$, we introduce the
overdamped relaxation time scale
\begin{equation}
T_0 := \frac{\zeta L_0^4}{\beta_{\kappa}}.
\end{equation}
We rescale time, velocities, chemical potential, and flux by
\begin{equation}
t=T_0\tilde t,\qquad
v=\frac{L_0}{T_0}\tilde v,\qquad
\mu=\frac{\beta_{\kappa}}{L_0^2}\tilde\mu,\qquad
J=\frac{L_0}{T_0}\tilde J,
\label{eq:dyn_scalings}
\end{equation}
and note that $\partial_s = L_0^{-1}\partial_{\tilde s}$. The shape variation is defined by
\begin{equation}
\delta E
=
\int_\Gamma \Big(F_t\,\xi + F_n\,\eta\Big)\,ds
+\int_\Gamma \mu\,\delta c\,ds,
\label{eq:shape_deriv_def_app}
\end{equation}
so the force densities satisfy $[F_{t,n}] = \mathrm{energy}/\mathrm{length}^2$ and admit the scaling
\begin{equation}
F_{t,n}=\frac{\beta_\kappa}{L_0^3}\,\tilde F_{t,n}.
\label{eq:force_scaling}
\end{equation}
The overdamped dynamics $\zeta v_{t,n}=-F_{t,n}$ then becomes, using \eqref{eq:dyn_scalings} and
$T_0=\zeta L_0^4/\beta_\kappa$,
\begin{equation}
\tilde v_t=-\tilde F_t,\qquad
\tilde v_n=-\tilde F_n,
\label{eq:overdamped_nd}
\end{equation}
i.e.\ the dimensionless damping coefficient is $\tilde\zeta\equiv 1$.\\

{\bf Cahn--Hilliard flux and mobility:}

We prescribe conserved Cahn--Hilliard dynamics along the physical arclength coordinate,
\begin{equation}
J
=\frac{L_0}{T_0}\tilde J 
=-\frac{\beta_{\kappa}}{L_0^2} M\,\partial_s\mu
= -M\,\partial_s\mu.
\label{eq:flux_dim_app}
\end{equation}
Substituting \eqref{eq:dyn_scalings} and $\partial_s=L_0^{-1}\partial_{\tilde s}$ into
\eqref{eq:flux_dim_app} yields
\begin{equation}
\tilde J = - (M\zeta)\,\partial_{\tilde s}\tilde\mu.
\end{equation}
We therefore define the dimensionless mobility
\begin{equation}
\tilde M := M\zeta,
\label{eq:mobility_nd}
\end{equation}
so that the nondimensional flux law takes the standard form
\begin{equation}
\tilde J = -\tilde M\,\partial_{\tilde s}\tilde\mu.
\label{eq:flux_nd_app}
\end{equation}
\noindent
In the main text we drop all tildes for notational simplicity. 

\section{Formal Sharp Interface Analysis}\label{secapp: sharp}

In this Appendix, we consider the case of an inextensible filament ($S^1$), and examine the effects of the constraints set first by the turning number and total mass conservation, and only afterwards by closure. Doing things in this order reveals an interesting mathematical structure that allows us to present a more formal proof of some results that were argued and used heuristically in the main text. 
 
Without the closure constraint, we have the following statement concerning the energetics subject to our first two constraints:

\begin{proposition}\label{prop: optimal_kappa}
Let $\Gamma$ be an inextensible closed curve of length $L:=\int_\Gamma ds$. Fix $c\in L^2(\Gamma)$ with total mass $C_0:=\int_\Gamma c\,ds$.
Consider
$
E[\kappa;c]:=\frac12\int_\Gamma (\kappa-\kappa_0 c)^2\,ds
$
minimized over $\kappa\in L^2(\Gamma)$ subject to the turning number constraint
$
\int_\Gamma \kappa\,ds = K_0.
$
Then the unique minimizer is
$$
\kappa^*(s)=\kappa_0 c(s)+\lambda,
\qquad
\lambda=\frac{K_0-\kappa_0 C_0}{L},
$$
and the minimal energy equals
$$
E[\kappa^*;c]
=\frac12\int_\Gamma (\kappa^*-\kappa_0 c)^2\,ds
=\frac12\,L\lambda^2
=\frac{(K_0-\kappa_0 C_0)^2}{2L},
$$
which depends only on the total mass $C_0$, the total curvature $K_0$, and the length $L$, but not on the spatial distribution of $c$.
\end{proposition}

\begin{proof}
Introduce the Lagrangian
\[
\mathcal{L}[\kappa]
:=\frac12\int_\Gamma(\kappa-\kappa_0 c)^2\,ds
-\lambda\left(\int_\Gamma\kappa\,ds-K_0\right).
\]
For any variation $\delta\kappa\in L^2(\Gamma)$,
\[
\delta\mathcal{L}
=\int_\Gamma \big[(\kappa-\kappa_0 c)-\lambda\big]\delta\kappa\,ds.
\]
Imposing the constraint gives
\[
K_0=\int_\Gamma \kappa\,ds
=\kappa_0\int_\Gamma c\,ds+\lambda\int_\Gamma ds
=\kappa_0 C_0+\lambda L,
\]
so the Lagrange multiplier is $\lambda=(K_0-\kappa_0 C_0)/L$, and the minimizer is $\kappa^*(s)=\kappa_0 c(s)+\lambda$.
Finally, $\kappa^*-\kappa_0 c\equiv \lambda$ is constant, hence
\[
E[\kappa^*;c]=\frac12\int_\Gamma \lambda^2\,ds=\frac12\,L\lambda^2.
\]
\end{proof}

After eliminating the curvature variable, the problem has the same structure as in the work of  ~\cite{modica1987gradient} addressing the sharp-interface limit for Cahn-Hilliard theory in a fixed Euclidean space. The only additional ingredient is the bending term, whose quadratic structure ensures the required lower semicontinuity, thus in the small $\varepsilon$ limit the coupled energy of the composition field and filament bending, which we denote $E_\varepsilon$, converges (more precisely, $\Gamma$-converges~\cite{modica1987gradient}) to a sharp-interface functional whose minimizers satisfy $c\in\{0,1\}$ almost everywhere and have piecewise constant curvature
\[
\kappa=
\begin{cases}
\lambda, & c=0,\\
\lambda+\kappa_0, & c=1,
\end{cases}
\]
with smooth (formally, $C^1$) junctions across interfaces.

To now address the closure condition, we first exclude kinks so that the limit curve is a $C^1$ concatenation of circular arcs.

\begin{proposition}[absence of kinks]\label{prop:no_kink}
In the sharp-interface regime, the loop has no kinks; equivalently, the tangent is continuous along $\Gamma$.
\end{proposition}

\begin{proof}
Let $(\mathbf x_\varepsilon,c_\varepsilon)$ be minimizers and $E_\varepsilon$ is the energy with parameter $\varepsilon$. By $\Gamma$-convergence, the minimal energies are uniformly bounded:
\[
\sup_{\varepsilon\downarrow0}E_\varepsilon[\mathbf x_\varepsilon,c_\varepsilon]\le C.
\]
Since $E_\varepsilon$ controls $\frac12\int_\Gamma(\kappa_\varepsilon-\kappa_0c_\varepsilon)^2\,ds$ and $c_\varepsilon$ is bounded, we obtain a uniform $L^2$ bound on curvature,
$\sup_{\varepsilon}\int_\Gamma \kappa_\varepsilon^2\,ds<\infty$.
A kink would create a jump in the tangent angle, hence a Dirac component in $\kappa=\partial_s\theta$, contradicting $\kappa\in L^2$. Therefore no kinks occur.
\end{proof}

From the preceding two results we conclude that the minimizer (without closure but with the turning number and composition constraints) is a smooth piecewise combination of circular arcs. 

Next we see under what conditions this optimal curve could be a closed curve. Consider $N$ circular arcs with curvatures $\left\{\kappa_i\right\}$ (possibly of mixed sign) and arc lengths $\left\{l_i\right\}$. Let the turning angles be $\Delta \theta_i=\kappa_i l_i$, and define the cumulative angles
$$
S_0=0, \quad S_i=\sum_{j=1}^i \Delta \theta_j .
$$
The displacement vector of each arc in the complex plane $z=x+i y$ is
$$
v_i\left(\Delta \theta_i\right)= \begin{cases}\frac{e^{i S_i}-e^{i S_{i-1}}}{i \kappa_i}=e^{i\left(S_{i-1}+\frac{\Delta \theta_i}{2}\right)} \frac{2 \sin \left(\Delta \theta_i / 2\right)}{\kappa_i}, & \kappa_i \neq 0, \\ e^{i S_{i-1}} l_i, & \kappa_i=0 .\end{cases}
$$
A smooth $C^1$ closed curve can be formed if and only if
$$
\sum_{i=1}^N \Delta \theta_i=2 \pi m, \quad \sum_{i=1}^N v_i=0
$$
where the first condition stems from the turning number $m$ and the second ensures positional closure.

We recall that the composition $c \in\{0,1\}$ is piecewise constant and the curvature takes only two values $\kappa_{+}, \kappa_{-}$ on the corresponding dense and dilute phases. Meanwhile the total arc-lengths of the two phases is fixed by the mass constraint:
$$
L_{+}>0, \quad L_{-}>0\,, \quad L_{+}+L_{-}=L\, \quad L_+/L = C_0\,.
$$
Consider an alternating decomposition of the loop into $N$ circular arcs with curvatures $\kappa_i \in\left\{\kappa_{+}, \kappa_{-}\right\}$ (alternating signs/phases), lengths $l_i$, and turning angles $\Delta \theta_i=\kappa_i l_i$, we have the following results:

\begin{proposition}
    For the domain number $N=2$, closure is not generic.
\end{proposition}
\begin{proof}
    For $N=2$ the closure conditions read
$$
\Delta \theta_1+\Delta \theta_2=2 \pi m, \quad v_1+v_2=0 .
$$
Writing $v_i=\left(e^{i S_i}-e^{i S_{i-1}}\right) /\left(i \kappa_i\right)$ with $S_1=\Delta \theta_1, S_2=\Delta \theta_1+\Delta \theta_2$, the second condition reduces to

$$
\frac{e^{i \Delta \theta_1}-1}{i \kappa_1}+\frac{e^{i 2 \pi m}-e^{i \Delta \theta_1}}{i \kappa_2}=0 \Longleftrightarrow\left(\frac{1}{\kappa_1}-\frac{1}{\kappa_2}\right)\left(e^{i \Delta \theta_1}-1\right)=0 .
$$
Hence either $\kappa_1=\kappa_2$ and then $\Delta \theta_1+\Delta \theta_2=2 \pi m$ (a degenerate "two-arc circle" compatibility), or $e^{i \Delta \theta_1}=1$ (so $\Delta \theta_1 \in 2 \pi \mathbb{Z}$ ) forcing $\Delta \theta_2 \in 2 \pi \mathbb{Z}$. These are non-generic algebraic coincidences depending on $\kappa_i$ and $l_i$. In particular, with distinct $\kappa_1 \neq \kappa_2$ and prescribed $L_{ \pm}$(hence prescribed $\Delta \theta_i$ up to partition), closure typically fails. Thus $N=2$ does not generally admit a solution.
\end{proof}
The previous result shows that with only two arcs closure of the filament requires exceptional algebraic coincidences. That is, for generic prescribed $L_{ \pm}$ and distinct curvatures $\kappa_{+} \neq \kappa_{-}$, closure typically fails. When $N$ increases, additional geometric degrees of freedom become available. In particular, for even $N \geq 4$ one can enforce closure by equally splitting each phase so that the displacement vectors within each phase form a regular polygon and hence sum to zero. We state this construction next.

\begin{proposition}
    For the domain number $N \ge 4$, an equal-angle/length construction gives closure.  
\end{proposition}
\begin{proof}
Let $N=2M$ with $M\ge2$ and alternate curvatures $\kappa_+,\kappa_-,\dots$.
Split each phase equally:
\[
l^{(+)}=\frac{L_+}{M},\quad \Delta\theta^{(+)}=\frac{\kappa_+L_+}{M},\qquad
l^{(-)}=\frac{L_-}{M},\quad \Delta\theta^{(-)}=\frac{\kappa_-L_-}{M}.
\]
Then the turning constraint $\sum_{i=1}^N\Delta\theta_i=2\pi m$ gives
\[
\Delta\theta^{(+)}+\Delta\theta^{(-)}=\frac{2\pi m}{M}.
\]
For the $M$ ``$+$'' arcs, all displacement vectors $v_i$ have the same modulus (since $\Delta\theta^{(+)}$ is fixed),
and the argument increases by the constant step $\Delta\theta^{(+)}+\Delta\theta^{(-)}=2\pi m/M$ from one $+$-arc to the next.
Hence $\{v_i\}_{+\text{ phase}}$ forms a rotated, scaled regular $M$-gon, so
\[
\sum_{+\text{ phase}} v_i=0.
\]
The same argument yields $\sum_{-\text{ phase}} v_i=0$, and therefore $\sum_{i=1}^N v_i=0$.
Together with Proposition~\ref{prop:no_kink}, the constructed loop is a $C^1$ closed curve.
\end{proof}

Importantly, the equal-split construction not only proves existence of closed loops for even $N \geq 4$, but also provides a symmetric reference configuration $\mathbf{\Delta}^{\mathrm{eq}}$. A natural next question is whether this reference configuration is an isolated solution or has finite codimension.  We answer this by introducing the closure map and applying the implicit function theorem.
\begin{proposition}[degeneracy or uniqueness of closed minimizer solutions]\label{prop:local_manifold_closure}
Let $\boldsymbol{\Delta}=(\Delta\theta_1,\dots,\Delta\theta_N)$ satisfy the two linear constraints of turning number and mass conservation (hence the admissible set has dimension $d=N-2$).  
Define the closure map
\[
\Phi(\boldsymbol{\Delta})=\sum_{i=1}^N v_i(\boldsymbol{\Delta})\in\mathbb{R}^2.
\]
Assume that at the equal-split configuration $\boldsymbol{\Delta}^{\mathrm{eq}}$ one has $\Phi(\boldsymbol{\Delta}^{\mathrm{eq}})=0$ and that the Jacobian
$D\Phi(\boldsymbol{\Delta}^{\mathrm{eq}})$ has (generic) rank $2$. Then, in a neighborhood of
$\boldsymbol{\Delta}^{\mathrm{eq}}$ the set of closed configurations
\[
\Phi^{-1}(0)=\{\boldsymbol{\Delta}:\Phi(\boldsymbol{\Delta})=0\}
\]
is a smooth manifold of codimension $2$ in the admissible set, hence
\[
\dim \Phi^{-1}(0)=d-2=N-4.
\]
In particular, the equal-split closure is locally isolated for $N=4$, while for even $N>4$ there is a $(N-4)$-dimensional family of nearby (generically non-equal-split) closed minimizer configurations.
\end{proposition}

\begin{proof}
For $N=4$, the equal-split arrangement yields pair cancellation within each phase:
$v_1+v_3=0$ and $v_2+v_4=0$, hence $\Phi(\boldsymbol{\Delta}^{\mathrm{eq}})=0$.
For $N=6$, the three vectors in each phase are separated by $120^\circ$ and therefore sum to zero within each phase, again giving $\Phi(\boldsymbol{\Delta}^{\mathrm{eq}})=0$.

Imposing the two linear constraints leaves $d=N-2$ independent parameters. Since
$D\Phi(\boldsymbol{\Delta}^{\mathrm{eq}})$ has generic rank $2$, the implicit function theorem implies that
$\Phi^{-1}(0)$ is a smooth submanifold of codimension $2$ in the admissible set. Therefore
$\dim \Phi^{-1}(0)=d-2=N-4$, which yields the stated consequences for $N=4$ and even $N>4$.
\end{proof}

\section{Additional details on numerical methods and parameters} \label{secapp: numerical}

\subsection{Phase diagram parameter sweep details}
For each parameter point $\left(x_c, y_c\right)$, in the parameter plane of interest (for example ($C,\alpha$)) we consider its $3 \times 3$ neighborhood in the discrete parameter grid. For every existing neighbor $\left(x_n, y_n\right)$ and for every metastable solution previously identified at that neighbor, we extract the final Fourier coefficients $\hat{u}_n(k)$ of the found morphology and use them to initialize a new simulation at the fixed center parameters $(x_c, y_c)$. This process is intended to track every metastable solution that is supported at $(x_c, y_c)$ by initializing simulations in each neighbouring metastable basin to see which basin of attraction they then end up in. It is also useful in overcoming any free energy barriers between structures that are classified within the same morphological phase ({\em i.e.}, with the same $N$ value).

Before each such extra minimization from neighboring structures, we enforce the relevant global constraints in spectral space: the zero mode is adjusted to match the prescribed total $C$, and the curvature-related spectral contribution is projected or rescaled to preserve the total curvature, while higher modes retain the neighbor's structure. The free energy is then minimized at $\left(x_c, y_c\right)$ from this constraint-compliant initialization. This procedure is iteratively repeated to construct updated phase diagrams, until no new morphological phases or transitions are detected. The ground state at each parameter is then identified, as the locally stable branch with the lowest free energy.

\begin{figure*}[t]
\centering
\includegraphics[width=\linewidth]{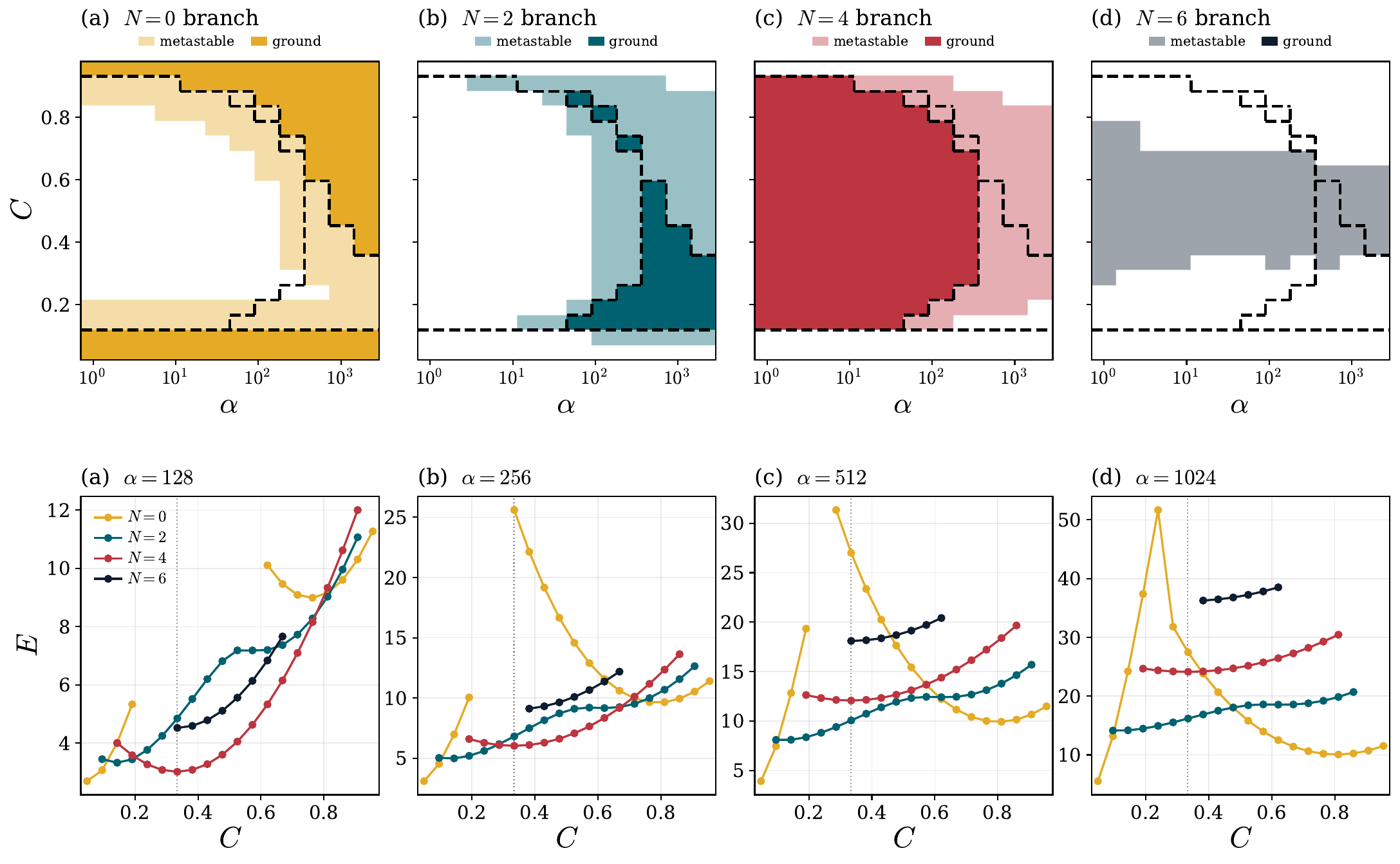}
\caption{
Metastable branch landscape on the ( $\alpha, C$ ) plane (figure \ref{fig:phase_ep05}) and energy of metastables at fixed $\alpha$, for $\varepsilon=0.05, \kappa_0=3, \beta=20, m=1$.
Top row [(a)-(d)] : branch-existence maps for interface number $N=0,2,4,6$. Within each panel, light colour marks parameter values at which the branch is a local minimum lying above the ground state, while the darker colour marks values at which the branch is itself the ground state. The dashed black lines is the phase boundary as in figure \ref{fig:energy_decomp}. The $N=6$ branch shown in (d) is widespread as a metastable basin yet is the ground state nowhere.
Bottom row [(a)-(d)]: 1D energy cuts at fixed $\alpha=128,256,512,1024$. Each curve plots the converged energy $E(C)$ of the interface number $N=0,2,4,6$. Markers are placed only at $C$ values at which that branch is realised as a local minimum. The vertical dotted line marks $C^{\star}=1 / \kappa_0 \approx 0.333$, the composition at which the curvature-composition mismatch of the dense phase is minimised as discussed in figure \ref{fig:energy_decomp}.
}
\label{fig:metastable_phasediagram}
\end{figure*}

\subsection{Gradient flow simulation: Details and invariants}
Nonlinear products in \eqref{eq:evolution_gkc}--\eqref{eq:vt_vn} induce aliasing in Fourier pseudo-spectral discretizations. We apply $3/2$-rule de-aliasing for nonlinear evaluations, together with a two-stage low-pass filtering of the normal velocity: we first use a 32 shape low-pass filter on the thermodynamic stress $\mathcal P$ \eqref{eq: P_stress}, which is empirically the most contamination-prone term. Then we apply a light 64-shape low pass filter to the assembled $v_n$ before advancing \eqref{eq:evolution_gkc}. Appropriate value of the filter can be tested on metastable states. 

We show that the three global constraints are exact invariants of the continuous gradient-flow dynamics~\eqref{eq:evolution_gkc}.\\

{\bf Turning number:}\label{sec:turning_cons}

For the turning number constraint \eqref{eq: turning}, using \eqref{eq:evolution_gkc} we have  
\begin{equation}
    \frac{d}{d t} \int_{S^1} \kappa h d \sigma=\int_0^{2 \pi} \partial_\sigma[\partial_s v_n+\kappa v_t] d \sigma=0
\end{equation}
In the Fourier pseudo-spectral scheme, $\partial_\sigma$ corresponds to multiplication by~$ik$; hence $\widehat{(\partial_\sigma f)}_0=0$ for any discrete field, and the trapezoidal rule is exact for band-limited functions. Hence the error comes with the time integration method.\\

{\bf Total mass:}\label{sec:mass_cons}

For the total mass constraint, from the evolution equations,
\begin{equation}
    \frac{d}{d t} \int_{S^1} ch d \sigma=\int_0^{2 \pi}\left(c_t h+c h_t\right) d \sigma=\int_0^{2 \pi} M \partial_\sigma\left(h^{-1} \partial_\sigma \mu\right) d \sigma=0
\end{equation}
and the same Fourier argument gives exact spatial conservation. The residual comes from the time integration method.\\

{\bf Closure:}\label{sec:closure_cons}
Define the closure vector $\mathbf{T}=(\cos\theta,\sin\theta)$ and $\mathbf{N}=(-\sin\theta,\cos\theta)$.  We have the time evolution
\begin{equation}
\frac{d}{dt}\int_{S^1}\mathbf{T}\,h\,d\sigma
=\int_{S^1}\bigl[(\omega(t)+\partial_s v_n+\kappa v_t)\,\mathbf{N}\,h
+(\partial_s v_t-\kappa v_n)\,\mathbf{T}\,h\bigr]\,d\sigma
=\int_{\Gamma}\left[\omega(t) \mathbf{N}+\partial_s\left(v_t \mathbf{T}+v_n \mathbf{N}\right)\right] d s
\label{eq:dXdt_raw}
\end{equation}
where $\omega$ is a rigid body uniform rotation rate~$\omega(t)$.  Since $
\int_{\Gamma} \partial_s(\cdot) d s=0
$
for any smooth field on a closed curve, the evolution of the closure integral can only arise from the spatially uniform rotation term. Consequently, if$
\int_{S^1} \mathbf{T} h d \sigma=0$ at $ t=0
$ then it remains zero for all $t \geq 0$.\\

{\bf Numerical error:}
At the continuous level, all three constraints are exact invariants. In the pseudo-spectral discretization, small residual violations arise from time discretization, spectral truncation of nonlinear terms, and the splitting of coupled variables. All contributions can be systematically reduced by using more advanced integration methods, decreasing $\Delta t$ or increasing the number of Fourier modes $N$, but this is not pursued here. In all simulations reported in this paper, the residuals of the turning number, total mass, and closure constraints remain below $10^{-4}$.

\subsection{Parameters values for figures}

\paragraph{Figure \ref{fig: metastables}}
$\alpha=1024, \beta=20, \varepsilon=0.05, \kappa_0=3, C=0.43, m=1$
spatial resolution $N_x = 256$, Fourier mode $N_{\rm mode}=64$.

\paragraph{Figure \ref{fig:closure_compare_combined}}
Both left and right panels use $\beta=20$, $\varepsilon=0.05$, $\kappa_0=3$, $m=1$, with spatial resolution $N_x=256$ and Fourier modes $N_{\rm mode}=64$.
\textbf{Left}: $\alpha=64$, $C=0.5$.
\textbf{Right}: $\alpha=1024$, $C=0.24$.

\paragraph{Figures  \ref{fig:phase_diagrams_alphaC},  \ref{fig:energy_decomp},  \ref{fig:phase_diagrams_k0C}}
All figures use $\beta=20, \varepsilon=0.05, \kappa_0=3, m=1$ and 
spatial resolution $N_x = 128$, Fourier mode $N_{\rm mode}=32$.
Optimization stop criteria: tolerance for termination by the norm of the Lagrangian gradient $10^{-6}$.
For the figure  \ref{fig:phase_diagrams_alphaC}, $C = C_{\mathrm{total}}/(2\pi)$, with $C_{\mathrm{total}}$ uniformly sampled from $0.3$ to $6.0$ in steps of $0.3$.

\paragraph{Figure \ref{fig:PDE_solution}}
Both left and right panels use $\beta=20$, $\varepsilon=0.05$, $\kappa_0=3$, $m=1$, motility $M=1$, spatial resolution $N_x=128$, and time step $\Delta t=10^{-7}$.
\textbf{Left}: $\alpha=64$, $C=0.43$.
\textbf{Right}: $\alpha=256$, $C=0.55$.

\paragraph{Figure \ref{fig:metastable_phasediagram}}
A small Gaussian perturbation of standard deviation $0.01$ is applied to the non-zero $c$-field Fourier modes of each warm-start initialization to break residual symmetry, so that the metastable branches identified are those recovered as local minima after this perturbation.

\end{document}

%% file: figures/curve_mapping.tex

\pgfmathsetmacro{\Rbase}{1.5} 
\pgfmathsetmacro{\aone}{0.24}
\pgfmathsetmacro{\atwo}{0.16}
\pgfmathsetmacro{\phione}{5}
\pgfmathsetmacro{\phitwo}{80}

\pgfmathsetmacro{\angTN}{210}
\pgfmathsetmacro{\rTN}{ \Rbase + \aone*cos(2*\angTN + \phione) + \atwo*sin(3*\angTN + \phitwo) }
\pgfmathsetmacro{\rTNp}{ -\aone*2*sin(2*\angTN + \phione) + \atwo*3*cos(3*\angTN + \phitwo) }
\coordinate (P) at ({\rTN*cos(\angTN)}, {\rTN*sin(\angTN)});
\pgfmathsetmacro{\dxTN}{ \rTNp*cos(\angTN) - \rTN*sin(\angTN) }
\pgfmathsetmacro{\dyTN}{ \rTNp*sin(\angTN) + \rTN*cos(\angTN) }
\pgfmathsetmacro{\lenTN}{ sqrt(\dxTN*\dxTN + \dyTN*\dyTN) }
\pgfmathsetmacro{\Tx}{ \dxTN/\lenTN }
\pgfmathsetmacro{\Ty}{ \dyTN/\lenTN }
\pgfmathsetmacro{\Nx}{ -\Ty }
\pgfmathsetmacro{\Ny}{  \Tx }

\pgfmathsetmacro{\angS}{150}
\pgfmathsetmacro{\wS}{20}
\pgfmathsetmacro{\angStart}{\angS - \wS/2}
\pgfmathsetmacro{\angEnd}{\angS + \wS/2}
\pgfmathsetmacro{\rMid}{ \Rbase + \aone*cos(2*\angS + \phione) + \atwo*sin(3*\angS + \phitwo) }
\coordinate (M) at ({\rMid*cos(\angS)}, {\rMid*sin(\angS)});
\coordinate (Mds) at ({\rMid*cos(\angS)-0.2}, {\rMid*sin(\angS)+0.1});
\pgfmathsetmacro{\rStart}{ \Rbase + \aone*cos(2*\angStart + \phione) + \atwo*sin(3*\angStart + \phitwo) }
\pgfmathsetmacro{\rEnd}{   \Rbase + \aone*cos(2*\angEnd   + \phione) + \atwo*sin(3*\angEnd   + \phitwo) }
\coordinate (A) at ({\rStart*cos(\angStart)}, {\rStart*sin(\angStart)});
\coordinate (B) at ({\rEnd*cos(\angEnd)},     {\rEnd*sin(\angEnd)});
\pgfmathsetmacro{\rSpA}{ -\aone*2*sin(2*\angStart + \phione) + \atwo*3*cos(3*\angStart + \phitwo) }
\pgfmathsetmacro{\dxA}{ \rSpA*cos(\angStart) - \rStart*sin(\angStart) }
\pgfmathsetmacro{\dyA}{ \rSpA*sin(\angStart) + \rStart*cos(\angStart) }
\pgfmathsetmacro{\lenA}{ sqrt(\dxA*\dxA + \dyA*\dyA) }
\pgfmathsetmacro{\nAx}{ -\dyA/\lenA }
\pgfmathsetmacro{\nAy}{  \dxA/\lenA }
\pgfmathsetmacro{\sA}{ (\nAy<0) ? -1 : 1 }
\pgfmathsetmacro{\NxA}{ \sA * \nAx }
\pgfmathsetmacro{\NyA}{ \sA * \nAy }
\pgfmathsetmacro{\rSpB}{ -\aone*2*sin(2*\angEnd + \phione) + \atwo*3*cos(3*\angEnd + \phitwo) }
\pgfmathsetmacro{\dxB}{ \rSpB*cos(\angEnd) - \rEnd*sin(\angEnd) }
\pgfmathsetmacro{\dyB}{ \rSpB*sin(\angEnd) + \rEnd*cos(\angEnd) }
\pgfmathsetmacro{\lenB}{ sqrt(\dxB*\dxB + \dyB*\dyB) }
\pgfmathsetmacro{\nBx}{ -\dyB/\lenB }
\pgfmathsetmacro{\nBy}{  \dxB/\lenB }
\pgfmathsetmacro{\sB}{ (\nBy<0) ? -1 : 1 }
\pgfmathsetmacro{\NxB}{ \sB * \nBx }
\pgfmathsetmacro{\NyB}{ \sB * \nBy }

\pgfmathsetmacro{\cx}{-5.6}
\pgfmathsetmacro{\cy}{0}
\pgfmathsetmacro{\rad}{1.2}
\coordinate (C) at ({\cx},{\cy});
\pgfmathsetmacro{\sigA}{80}
\pgfmathsetmacro{\dsig}{20}
\pgfmathsetmacro{\sigB}{\sigA + \dsig}
\pgfmathsetmacro{\sigM}{\sigA + \dsig/2}
\coordinate (Sstart) at ({\cx + \rad*cos(\sigA)}, {\cy + \rad*sin(\sigA)});
\coordinate (Send)   at ({\cx + \rad*cos(\sigB)}, {\cy + \rad*sin(\sigB)});
\coordinate (Smid)   at ({\cx + \rad*cos(\sigM)}, {\cy + \rad*sin(\sigM)});
\pgfmathsetmacro{\sigZero}{310}
\coordinate (Szero) at ({\cx + \rad*cos(\sigZero)}, {\cy + \rad*sin(\sigZero)});

\pgfmathsetmacro{\tick}{0.22}
\pgfmathsetmacro{\rtick}{0.18}


\draw[line width=1pt, domain=0:360, samples=400, smooth, variable=\ang]
  plot ({ ( \Rbase + \aone*cos(2*\ang + \phione) + \atwo*sin(3*\ang + \phitwo) ) * cos(\ang) },
        { ( \Rbase + \aone*cos(2*\ang + \phione) + \atwo*sin(3*\ang + \phitwo) ) * sin(\ang) });
\node[anchor=north] at ($ (0.2,-\rad-0.3) $) {$\mathbf{x} \in \Gamma$};
\def\boxxmin{-2.7}
\def\boxxmax{2.6}
\def\boxymin{-2.2}
\def\boxymax{2.2}

\draw[dashed, thick] (\boxxmin,\boxymin) rectangle (\boxxmax,\boxymax);

\node[anchor=south east] at (\boxxmax, \boxymin) {$\mathbb{R}^2$};

\fill (P) circle (1.2pt);
\tikzset{vec/.style={->, line width=1pt}}
\draw[vec] (P) -- ++({1.2*\Tx},{1.2*\Ty}) node[left=2pt] {$\mathbf T$};
\draw[vec] (P) -- ++({1.2*\Nx},{1.2*\Ny}) node[below=2pt]  {$\mathbf N$};
\node[anchor=south east] at ($(P)+(0.1, -0.6)$) {$\mathbf{x}(\sigma_0, t)$};

\draw[line width=1.8pt, domain=\angStart:\angEnd, samples=120, smooth, variable=\ang]
  plot ({ ( \Rbase + \aone*cos(2*\ang + \phione) + \atwo*sin(3*\ang + \phitwo) ) * cos(\ang) },
        { ( \Rbase + \aone*cos(2*\ang + \phione) + \atwo*sin(3*\ang + \phitwo) ) * sin(\ang) });
\draw[line width=1pt] (A) -- ++({\tick*\NxA},{\tick*\NyA});
\draw[line width=1pt] (B) -- ++({\tick*\NxB},{\tick*\NyB});
\node[anchor=south, yshift=3pt] at (Mds) {$ds$};

\draw[line width=1pt] (C) circle (\rad);
\node[anchor=north] at ($ (C) + (0,-\rad-0.25) $) {$\sigma\in S^1$};
\draw[line width=1.6pt] (Sstart) arc[start angle=\sigA, end angle=\sigB, radius=\rad];
\draw[line width=1pt] (Sstart) -- ++({\rtick*cos(\sigA)},{\rtick*sin(\sigA)});
\draw[line width=1pt] (Send)   -- ++({\rtick*cos(\sigB)},{\rtick*sin(\sigB)});
\node[anchor=center, yshift=3pt] at ({\cx + (\rad+0.15)*cos(\sigM)}, {\cy + (\rad+0.15)*sin(\sigM)}) {$d\sigma$};

\fill (Szero) circle (1.1pt) node[below=2pt] {$\sigma_0$};
\draw[densely dashed, ->, line width=0.9pt] (Szero) -- (P);

\draw[densely dashed, ->, line width=0.9pt] (Smid) -- (M)
  node[midway, above=3pt] {$ds = h\, d\sigma$};